\documentclass[12pt]{article}

\usepackage{longtable, ltcaption}%
\usepackage{amsmath,amsthm,amssymb}
\usepackage[margin=1in]{geometry}
\usepackage{hyperref}
\hypersetup{
    colorlinks=true,
    linkcolor=blue,
    citecolor=blue,
    urlcolor=blue,
    filecolor=blue
}
\usepackage{setspace,placeins,booktabs,ctable}
\usepackage{cleveref, graphicx, indentfirst,ragged2e}
\usepackage[utf8]{inputenc}
\usepackage{bm}
\usepackage{comment}
\usepackage[normalem]{ulem}

\usepackage{caption, subcaption}
\usepackage{enumitem}
\setlist{itemsep=2pt, topsep=0pt, parsep=0pt, partopsep=0pt}

\usepackage{footmisc}
\newenvironment{marginalign}{%
  \refstepcounter{equation}%
  \align\tag*{\makebox[0pt][l]{(\theequation)}}%
}{%
  \endalign
}

\makeatletter
\def\thm@space@setup{%
  \thm@preskip=5pt    %
  \thm@postskip=0pt   %
}
\makeatother

\usepackage[bibstyle=authoryear, citestyle=authoryear, doi=false, url=true, backend=biber, maxbibnames=50, maxcitenames=2, uniquelist=false, uniquename=false, sorting=nyt]{biblatex}
\renewbibmacro{in:}{}
\DeclareNameAlias{default}{family-given/given-family}
\newcommand{\citet}{\textcite}
\defbibheading{bibliography}[\refname]{%
  \vspace{5pt}
  \section*{#1}
  \vspace{3pt}
}
\defbibheading{bibliography}[\refname]{%
  \section*{#1}
}

\newcommand{\indicator}[1]{\mathbf{1}\{#1\}}

    \def\independenT#1#2{\mathrel{\setbox0\hbox{$#1#2$}%
    \copy0\kern-\wd0\mkern4mu\box0}}
\newcommand{\E}{\mathbb{E}}
\renewcommand{\L}{\mathrm{L}}
\renewcommand{\P}{\mathrm{P}}

\newcommand{\T}{\text{T}}
\newcommand{\tightminus}{\mspace{0mu}\scalebox{1.5}[1.0]{-}\mspace{0mu}}
\newcommand{\tightequals}{\mspace{0mu}\scalebox{1}[1.0]{=}\mspace{0mu}}
\newcommand{\smallminus}{\mathbin{\!-\!}}

\newcommand{\smallequals}{\mathrel{\!=\!}}

\newtheorem{theorem}{Theorem}

\newtheorem{lemma}{Lemma}

\newtheorem{assumption}{Assumption}

\newtheorem{proposition}{Proposition}

\newtheorem{inneruassumption}{Assumption}
\newenvironment{namedassumption}[1]
  {\renewcommand\theinneruassumption{#1}\inneruassumption}
  {\endinneruassumption}

\crefname{figure}{Figure}{Figures}
\crefname{assumption}{Assumption}{Assumptions}
\crefname{inneruassumption}{Assumption}{Assumptions}

\theoremstyle{definition}

\newtheorem{remark}{Remark}

\usepackage{titlesec}
\titleformat{\section}
  {\Large\bfseries} %
  {\thesection} %
  {0.5em} %
  {} %
\titlespacing*{\section}
  {0pt} %
  {0pt} %
  {0pt} %
\titleformat{\subsection}
  {\large\bfseries} %
  {\thesubsection} %
  {1em} %
  {} %
\titlespacing*{\subsection}
  {0pt} %
  {5pt} %
  {0pt} %
\titlespacing*{\subsubsection}
  {0pt} %
  {5pt} %
  {0pt} %
\titlespacing*{\paragraph}
  {0pt}   %
  {5pt}   %
  {5pt}   %

\usepackage{mathptmx} %

\usepackage[nolists,tablesfirst]{endfloat} %

\usepackage{xr}
\title{\textbf{Difference-in-Differences when Parallel Trends Holds Conditional on Covariates\footnote{Some of the results in this paper were originally in ``Difference-in-differences with time-varying covariates'' (\citet{caetano-callaway-payne-santanna-2022}).  This paper and our companion paper, ``Difference-in-differences with bad controls'' (Caetano, Callaway, Payne, and Rodrigues (2025)), replace that paper.  The code for the new estimation approaches proposed in the paper is provided in the \texttt{R ptetools} package, which is available on CRAN.  Code for the TWFE and AIPW diagnostics discussed in the paper is available in the \texttt{R twfeweights} package, which is available at \url{https://github.com/bcallaway11/twfeweights}.  We thank Kyle Butts, Andrew Goodman-Bacon, Pedro Sant'Anna, Tymon Sloczynski, as well as many seminar and conference participants for helpful comments.%
}}}

\author{Carolina Caetano\footnote{John Munro Godfrey, Sr.\ Department of Economics, University of Georgia.  \href{mailto:carol.caetano@uga.edu}{carol.caetano@uga.edu}} \and Brantly Callaway\footnote{John Munro Godfrey, Sr.\ Department of Economics, University of Georgia.  \href{mailto:brantly.callaway@uga.edu}{brantly.callaway@uga.edu}}}

\begin{document}

\maketitle

\begin{abstract}
    We consider difference-in-differences identification and estimation strategies when the parallel trends assumption holds conditional on covariates, which can be time-varying, time-invariant, or both. We uncover several weaknesses of two-way fixed effects (TWFE) regressions in this context. The most important, which we call \textit{hidden linearity bias}, arises because transformations that eliminate unit fixed effects also transform the covariates, either implicitly changing the identification strategy or relying on correct model specification. We provide diagnostics for assessing a TWFE regression's susceptibility to hidden linearity bias and propose alternative estimation strategies that circumvent these issues.
\end{abstract}

\bigskip

\bigskip

\noindent {\bfseries {JEL Codes:}} C14, C21, C23

\bigskip

\noindent {\bfseries {Keywords:}}  Difference-in-Differences, Time-Varying Covariates,  Time-Invariant Covariates, Hidden Linearity Bias, Two-way Fixed Effects Regression, Doubly Robust Estimation, Conditional Parallel Trends, Treatment Effect Heterogeneity

\vspace{10em}

\pagebreak

\doublespacing
\setlength{\jot}{2pt} %
\setlength{\abovedisplayskip}{5pt} %
\setlength{\belowdisplayskip}{4pt} %
\allowdisplaybreaks %

\section{Introduction}

In this paper, we study difference-in-differences (DiD) when the parallel trends assumption holds after conditioning on covariates.  Researchers often include covariates in the parallel trends assumption to make it more plausible (e.g., \citet{heckman-ichimura-smith-todd-1998,abadie-2005}), with the aim of comparing the change in outcomes over time for treated units to the change in outcomes over time for untreated units with similar observed characteristics.
We pay careful attention to the different types of covariates that can show up in the parallel trends assumption: time-varying covariates and/or time-invariant covariates.  The econometrics literature and empirical applications often use covariates in substantially different ways.  In the econometrics literature, it is common to assume that the covariates are all time-invariant or, if there are time-varying covariates, to use a pre-treatment value of the time-varying covariates as a time-invariant covariate; see, for example, \citet{abadie-2005,bonhomme-sauder-2011,santanna-zhao-2020,callaway-santanna-2021}.  On the other hand, in empirical work, the most common way to include covariates is in the following TWFE regression\footnote{Our discussion of limitations of TWFE regressions is specific to this particular TWFE regression.  As emphasized in \citet{wooldridge-2025}, many of the drawbacks of common versions of TWFE regressions can be addressed by using more flexible TWFE regressions (e.g., including many additional interaction terms).  Using a more flexible regression could address some of the issues that we highlight below.  We focus on this particular TWFE regression because it is very common in empirical work.  \pagebreak In Supplementary Appendix \ref{sa:applications-review}, we systematically reviewed 25 empirical difference-in-differences papers.  Of these, 15 included covariates in the same manner as in \Cref{eqn:twfe}, 7 did not include time-varying covariates (6 of these did not include any covariates at all, 1 included only time-invariant covariates), 2 included baseline versions of the covariates (i.e., the values of the time-varying covariates in the first period), and 1 included both time-varying covariates as in \Cref{eqn:twfe} and baseline versions of the covariates. In addition, the TWFE regression in \Cref{eqn:twfe} is also emphasized as the way to introduce covariates in DiD applications in textbook treatments of difference-in-differences (e.g., \citet[Sections 6.5.2 and 10.6.4]{wooldridge-2010}, \citet[Chapter 5]{angrist-pischke-2008}, \citet[Chapter 9]{cunningham-2021}).}
\begin{align} \label{eqn:twfe}
    Y_{it} = \theta_t + \eta_i + \alpha D_{it} + X_{it}'\beta + e_{it},
\end{align}
where $\theta_t$ is a time fixed effect, $\eta_i$ is an individual fixed effect, $D_{it}$ is a binary treatment indicator, and $X_{it}$ are time-varying covariates.   $\alpha$ is the coefficient of interest in this regression, sometimes interpreted as ``the causal effect of the treatment'' or, in the presence of treatment effect heterogeneity, often loosely interpreted as some kind of average treatment effect parameter. Being able to include covariates is one of the original main attractions of using a TWFE regression to implement a DiD identification strategy.  For example, \citet{angrist-pischke-2008} write: ``A second advantage of regression-DD is that it facilitates empirical work with regressors.''  Relative to the econometrics literature discussed above, one immediately noticeable difference with the TWFE regression is that it does not include time-invariant covariates. This is because, if time-invariant covariates enter the TWFE model in an analogous way to the time-varying covariates (i.e., with a time-invariant coefficient), then they will be absorbed into the unit fixed effect.  This is a reason commonly given for not including time-invariant covariates in difference-in-differences applications.

We uncover several limitations of this TWFE regression.  Although TWFE regressions with only two time periods are known to be robust to treatment effect heterogeneity under unconditional parallel trends, we show that  TWFE regressions that rely on conditional parallel trends assumptions are susceptible to a number of problems \textit{even in the case with only two time periods}. We show that TWFE regressions can include non-neglible misspecification bias terms %
for any of three reasons: (1) violations of certain linearity conditions on the model for untreated potential outcomes over time, (2) paths of untreated potential outcomes that depend on the level of time-varying covariates in addition to (or instead of) the change in the covariates over time, and (3) paths of untreated potential outcomes that depend on time-invariant covariates.  The first issue is expected, as similar conditions show up in the literature on interpreting cross-sectional regressions under unconfoundedness (\citet{goldsmith-hull-kolesar-2024,blandhol-bonney-mogstad-torgovitsky-2025,hahn-2023}). Also, assuming a linear model for untreated potential outcomes is often a key step for motivating linear models for the outcome itself (see \citet{angrist-pischke-2008} for a number of examples).

Issues (2) and (3) are more subtle. We refer to the bias arising from them as \textit{hidden linearity bias}.  Unlike Issue (1), there is no analog of hidden linearity bias in cross-sectional settings. This bias arises because, to estimate \Cref{eqn:twfe}, the unit fixed effect is removed by transforming the model.  For example, with two time periods, $\alpha$ is estimated by taking first differences to eliminate the unit fixed effect.  A consequence is that the covariates are also differenced, so the estimated model ultimately only controls for \emph{changes} in the time-varying covariates. Below, we consider an application with state-level panel data, a state-level treatment, and time-varying covariates such as a state's population.  In this context, controlling for the change in a state's population may not account for the level of a state's population (i.e., states with similar changes in population could have quite different levels of population), or for other time-invariant covariates such as region.\footnote{It is also common to motivate the TWFE regression in \Cref{eqn:twfe} by a selection-on-observables argument conditional on covariates and unit and time fixed effects.  Versions of the same issues we highlight here also apply in this case, as these approaches require the model to be correctly specified.\label{fn:selection-on-observables}}

An important question is how much the bias discussed above matters in practice.  Since this bias is hard to measure directly, we propose simple diagnostic tools to assess TWFE regressions' sensitivity to hidden linearity bias.  Our idea is to recast $\alpha$ from the TWFE regression as a re-weighting estimator (see \citet{aronow-samii-2016,chattopadhyay-zubizarreta-2023} for related ideas).  We recover the ``implicit regression weights'' (which are straightforward to calculate) and apply them to the levels of time-varying covariates and time-invariant covariates.  If the implicit regression weights balance these covariates, hidden linearity bias is likely small. Otherwise, $\alpha$ may be quite sensitive to violations of linearity conditions.  Moreover, even if none of these issues arise, TWFE regressions deliver weighted averages of conditional ATTs with non-transparent weights that can be negative or exhibit weight-reversal properties similar to the ones pointed out in \citet{sloczynski-2022} in the cross-sectional case.

We propose several new estimation strategies that do not suffer from \textit{any} of the limitations of the TWFE regression discussed above.  Essentially, in line with the parallel trends assumption, we first difference the outcome, but we do not apply the same transformation to the covariates, allowing the levels of time-varying covariates and time-invariant covariates to remain in the estimating equation. We operationalize this idea by building on popular estimators in the DiD literature, particularly the augmented inverse propensity score weighting (AIPW) estimators in \citet{santanna-zhao-2020,callaway-santanna-2021}, though the same idea could be used with other estimators such as matching, inverse probability tilting (\citet{graham-pinto-egel-2012}), entropy balancing (\citet{hainmueller-2012}), and covariate balancing propensity score (\citet{imai-ratkovic-2014}).
We also show that AIPW estimators of the ATT under conditional parallel trends can be reformulated as re-weighting estimators. Combined with our TWFE diagnostics, this allows researchers to compare covariate balancing properties across leading ATT estimators during the ``design phase,'' before using the outcome at all (\citet{ho-imai-king-stuart-2007,rubin-2008,imbens-rubin-2015}).

We conclude the paper by revisiting an application from \citet{cheng-hoekstra-2013} on the effects of stand-your-ground laws on homicides.  We find that including time-varying covariates, such as a state's population and/or median income, in a TWFE regression balances the average of the within-transformed covariates but often does little to improve (and in some cases worsens) covariate balance in terms of the levels of the same covariates or in terms of time-invariant covariates.  Using our approach, covariate balance is substantially improved.  Our estimates are mostly qualitatively similar to those reported in \citet{cheng-hoekstra-2013}, though we find somewhat weaker evidence that stand-your-ground laws increase homicides.

The paper proceeds as follows.  \Cref{sec:two-periods-setup} presents the main assumptions and some preliminary identification results.  \Cref{sec:twfe} discusses the limitations of TWFE regressions under conditional parallel trends, and  \Cref{sec:diagnostics} proposes simple diagnostics for assessing the extent of hidden linearity bias in TWFE regressions.  These sections focus on the baseline two-period, two-group DiD setup, where many of our main insights can already be seen.  \Cref{sec:multiple-periods} extends the arguments to settings with multiple periods and variation in treatment timing across units.  \Cref{sec:estimation} proposes alternative estimation strategies that circumvent TWFE regressions' limitations. \Cref{sec:application} revisits an application from \citet{cheng-hoekstra-2013} on the effects of stand-your-ground laws on homicides. Finally, \Cref{sec:conclusion} concludes.

\section{Setup} \label{sec:two-periods-setup}

 For much of the paper, we consider the canonical two-period setting, with periods denoted by $t^*$ and $t^*\smallminus 1,$ and no treated units in the first period.  Let $D_i$ be a binary treatment indicator, whose time subscript is omitted because treatment differences occur only in the second period.  Let $X_{it}$ denote a $k \times 1$ vector of time-varying covariates for unit $i$ in time period $t$, and $Z_i$ an $l \times 1$ vector of time-invariant covariates.  Let $Y_{it}$ denote the observed outcome, with $Y_{it}(1)$ and $Y_{it}(0)$ the treated and untreated potential outcomes.  Observed and potential outcomes are related by $Y_{it^*} = D_i Y_{it^*}(1) + (1\smallminus D_i)Y_{it^*}(0)$ and $Y_{it^*\tightminus 1} = Y_{it^*\tightminus 1}(0)$;  i.e., in the second period we observe treated (untreated) potential outcomes for treated (untreated) units, and in the first period we observe untreated potential outcomes for all units.\footnote{The discussion above implicitly imposes a SUTVA assumption and a no-anticipation assumption (that pre-treatment outcomes are not affected by eventually participating in the treatment).  These are standard in the DiD literature.  We discuss no-anticipation in more detail in \Cref{sec:multiple-periods}.} We also suppose throughout the paper that all expectations exist and take all statements conditional on covariates to hold almost surely.

\subsection{Identification}

Following the vast majority of the difference-in-differences literature, we target identifying the average treatment effect on the treated (ATT), which is given by
\begin{align*}
    \mathrm{ATT} := \E[Y_{t^*}(1) - Y_{t^*}(0) | D=1].
\end{align*}
We also define a conditional-on-covariates version of the ATT as
\begin{align*}
    \mathrm{ATT}(X_{t^*}, X_{t^*-1}, Z) := \E[Y_{t^*}(1) - Y_{t^*}(0) | X_{t^*}, X_{t^*-1}, Z, D=1].
\end{align*}
We make the following assumptions:
\begin{assumption}[Random Sampling] \label{ass:sampling} The observed data $\{Y_{it^*}, Y_{it^*-1}, X_{it^*}, X_{it^*-1}, Z_i, D_i \}_{i=1}^n$ are i.i.d.
\end{assumption}

\begin{assumption}[Overlap] \label{ass:overlap} $\P(D=1) > \epsilon$ and $\P(D=1|X_{t^*},X_{t^*-1},Z) < 1-\epsilon$ for some $\epsilon>0$.
\end{assumption}

\begin{assumption}[Conditional Parallel Trends]  \label{ass:conditional-parallel-trends}
\begin{align*}
    \E[\Delta Y_{t^*}(0) | X_{t^*}, X_{t^*-1}, Z, D=1] = \E[\Delta Y_{t^*}(0) | X_{t^*}, X_{t^*-1}, Z, D=0] .
\end{align*}
\end{assumption}

\Cref{ass:sampling} assumes we have access to an iid two-period panel. \Cref{ass:overlap} is a standard version of an overlap condition often invoked in the DiD literature (e.g., \citet{abadie-2005}). In practice, it says that, for all treated units, there exist untreated units with the same characteristics. \Cref{ass:conditional-parallel-trends} says that the path of untreated potential outcomes is the same on average for treated and untreated groups after conditioning on time-varying and time-invariant covariates.  \Cref{ass:conditional-parallel-trends} also implicitly restricts the covariates to be unaffected by the treatment, ruling out so-called ``bad controls.''  See Supplementary Appendix \ref{sa:no-bad-controls} and \citet{caetano-callaway-payne-santanna-2022} for more details.

Under \Cref{ass:sampling,ass:conditional-parallel-trends,ass:overlap}, the ATT is identified, and, in particular, it is given by
\begin{align} \label{eqn:att}
    \mathrm{ATT} = \E[\Delta Y_{t^*} | D=1] - \E\Big[ \E[\Delta Y_{t^*} | X_{t^*}, X_{t^*-1}, Z, D=0] \Big| D=1 \Big].
\end{align}
We state this result formally in \Cref{prop:att} in the Supplementary Appendix. It follows from the same arguments as in existing work on difference-in-differences such as \citet{heckman-ichimura-smith-todd-1998,abadie-2005}, up to separately keeping track of the time-varying and time-invariant covariates. This expression says that the ATT can be recovered by comparing the mean path of outcomes for the treated group relative to the path they would have experienced if they had remained untreated.  Under conditional parallel trends, the counterfactual path can be recovered by taking the path of outcomes conditional on time-varying and time-invariant covariates for the untreated group and averaging over the covariate distribution of the treated group.

\section{Interpreting TWFE Regressions} \label{sec:twfe}

This section considers how to interpret $\alpha$ in the TWFE regression in \Cref{eqn:twfe}.  We continue to focus on the setting with two time periods where no units are treated in the first time period and some, but not all, units become treated in the second time period.  This is a favorable setting for TWFE regressions as it does not introduce problems related to using already-treated units in the comparison group (\citet{chaisemartin-dhaultfoeuille-2020,goodman-bacon-2021}).

For interpreting the TWFE regression, many of our results involve linear projections.  Let $\L(D|\Delta X_{t^*})$ denote the (population) linear projection of $D$ on $\Delta X_{t^*}$.\footnote{All of the linear projections in this section include an intercept. This involves a slight abuse of notation where, for example, we augment $\Delta X_{t^*}$ so that it includes an intercept in addition to the change in time-varying covariates over time. Similarly, we also slightly abuse notation in \Cref{eqn:fd} by taking $\beta$ to include an extra parameter in its first element corresponding to the intercept.  For all results below involving linear projections, we assume that they are well-defined.  This typically involves a rank condition, such as that $\E[\Delta X_{t^*} \Delta X_{t^*}']$ is positive definite. The vast majority of our results are provided in terms of population (rather than sample) quantities.  For expressions that only involve means and linear projections (which applies to many of our results below), analogous results hold for the corresponding sample quantities.}  That is,
\begin{align*}
    \L(D|\Delta X_{t^*}) := \Delta X_{t^*}' \E[\Delta X_{t^*} \Delta X_{t^*}']^{-1} \E[\Delta X_{t^*} D] = \Delta X_{t^*}'\gamma.
\end{align*}
Similarly, for $d \in \{0,1\}$, define
\begin{align*}
    \L_d(\Delta Y_{t^*}|\Delta X_{t^*}) := \Delta X_{t^*}'\E[\Delta X_{t^*} \Delta X_{t^*}' | D=d]^{-1} \E[\Delta X_{t^*} \Delta Y_{t^*} | D=d] = \Delta X_{t^*}' \beta_d,
\end{align*}
which is the linear projection of $\Delta Y_{t^*}$ on $\Delta X_{t^*}$ for the treated group (when $d=1$) and for the untreated group (when $d=0$), respectively.

Notice that, with two periods, it is helpful to equivalently re-write \Cref{eqn:twfe} as
\begin{equation} \label{eqn:fd}
  \Delta Y_{it^*} = \alpha D_i + \Delta X_{it^*}'\beta + \Delta e_{it^*}.
\end{equation}
We view \Cref{eqn:fd} as a linear projection model rather than as a linear conditional expectation/structural model, thus allowing for heterogeneous treatment effects.  Our interest in this section is in determining what kind of conditions are required to interpret $\alpha$ as the ATT or at least as a weighted average of some underlying treatment effect parameters.

Next, we provide our first main result on interpreting $\alpha$ in terms of underlying causal effect parameters along with some additional bias terms.
\begin{theorem}\label{thm:twfe} Under \Cref{ass:sampling,ass:overlap,ass:conditional-parallel-trends}, $\alpha$ from \Cref{eqn:fd} can be expressed as
    \begin{align*}
        \alpha &= \E\Big[ w(\Delta X_{t^*}) \mathrm{ATT}(X_{t^*},X_{t^*-1},Z) \Big| D=1 \Big] \label{eqn:thm-twfe-alp-a}\\
        &+ \E\Big[ w(\Delta X_{t^*}) \Big\{ \Big( \E[\Delta Y_{t^*} | X_{t^*}, X_{t^*-1}, Z, D=0] - \E[\Delta Y_{t^*} | X_{t^*}, X_{t^*-1}, D=0] \Big) \tag{A} \\
        & \hspace{100pt} + \Big(\E[\Delta Y_{t^*} | X_{t^*}, X_{t^*-1}, D=0] - \E[\Delta Y_{t^*} | \Delta X_{t^*}, D=0] \Big) \tag{B} \label{eqn:thm-twfe-alp-b} \\
        & \hspace{100pt} + \Big( \E[\Delta Y_{t^*} | \Delta X_{t^*}, D=0] - \L_0(\Delta Y_{t^*} | \Delta X_{t^*})\Big) \Big\} \Big| D=1 \Big] \tag{C}, \label{eqn:thm-twfe-alp-c}
    \end{align*}
    where $\displaystyle
        w(\Delta X_{t^*}) := \frac{1-\L(D | \Delta X_{t^*})}{\E[(1-\L(D|\Delta X_{t^*})) | D=1]},\;\;$  and  $\displaystyle \;\;\E[w(\Delta X_{t^*}) | D=1] = 1.$
\end{theorem}

\medskip
 \Cref{thm:twfe} states that $\alpha$ is a weighted average of underlying conditional ATTs (we discuss the weights in more detail below) plus several undesirable bias terms.\footnote{The proof of \Cref{thm:twfe} (especially the parts concerning the weights) is mechanically related to work on interpreting cross-sectional regressions under unconfoundedness or other related settings (\citet{angrist-1998,aronow-samii-2016,sloczynski-2022,goldsmith-hull-kolesar-2024,blandhol-bonney-mogstad-torgovitsky-2025,hahn-2023}).  The hidden linearity bias terms in the $\alpha$ expression (discussed in detail below) are specific to the DiD setting.} Note that the weights $w(\Delta X_{t^*})$ can be negative if there exist values of $\Delta X_{t^*}$ among the treated group such that $\L(D|\Delta X_{t^*}) > 1$. Since $w(\Delta X_{t^*})$ has mean one, the bias terms should be a first-order concern for empirical researchers.  This differs from several recent papers on interpreting regressions in different contexts, where the regression coefficient ends up including bias terms but the weights have mean zero (\citet{sun-abraham-2021,chaisemartin-dhaultfoeuille-2023b,goldsmith-hull-kolesar-2024}).

The bias in Term (A) arises because the regression in \Cref{eqn:fd} does not include time-invariant covariates.   This term suggests that failing to include time-invariant covariates in the TWFE regression when the path of untreated potential outcomes actually depends on time-invariant covariates undesirably contributes to how $\alpha$ is calculated.  In our application to stand-your-ground laws, this bias term could arise because of state-level time-invariant covariates that affect the path of untreated potential outcomes (e.g., region indicators, if trends in homicides (absent the policy) are different across different regions of the country).\footnote{\label{fn:region-year}As a point of clarification, it is common in empirical difference-in-differences applications that use state-level data to include region-time fixed effects (i.e., to include a region indicator with a time-varying coefficient).  This partially, though not entirely, addresses the issues discussed in this section; see, in particular, \citet[Section 5.2]{wooldridge-2025}. %
Perhaps a better example comes from DiD applications that use individual-level data, where it is less common to include time-invariant covariates with time-varying coefficients in the TWFE regression.  For example, it is uncommon in labor economics to include a person's race as a covariate in a TWFE regression because it does not vary over time, despite the fact that it seems likely that the path of many labor market outcomes depends on race.  Similarly, it is possible, but not common, to include a time-invariant continuous covariate with a time-varying coefficient in a TWFE regression.}

Term (B) is nonzero when the path of untreated potential outcomes depends on the levels of time-varying covariates instead of only on the change in covariates over time.  Together, we refer to Terms (A) and (B) as \textit{hidden linearity bias}.  Hidden linearity bias arises because an additional implication of linearity is that the estimating equation, as a by-product of differencing out the unit fixed effect, only ends up including the change in covariates over time.  In the case where the model is correctly specified, then this transformation of the covariates is appropriate. However, if the model is viewed as an approximation, then an undesirable implication of linearity is that it implicitly changes the identification strategy from one that includes levels of time-varying covariates and time-invariant covariates to one that only includes the change in the covariates over time.  For example, when including state population in a TWFE regression, the researcher likely intends to compare treated and untreated states with similar population levels. However, the TWFE regression effectively compares states with similar population changes, which, of course, could have quite different population levels.

Hidden linearity bias does not show up in cross-sectional settings.  Modern empirical work often views linear regressions as approximations.  In cross-sectional settings, the approximation view can be attractive as linearity itself may be a strong assumption, but using the linear model in estimation is convenient and has other good properties, such as being the best linear approximation to a possibly nonlinear conditional expectation function.  In contrast, the discussion above highlights that linearity has substantially more bite for DiD.  This suggests that linearity should be more carefully examined in panel data settings than in cross-sectional settings.

Term (C) is non-zero when the conditional expectation of the change in untreated potential outcomes conditional on covariates is nonlinear in the change in covariates over time.  This type of linearity condition is the one that researchers would likely suspect to be implicit in the TWFE regression.  A similar term shows up in cross-sectional settings with different papers discussing various conditions under which it is equal to zero (\citet{angrist-1998,blandhol-bonney-mogstad-torgovitsky-2025,hahn-2023}).  In general, this term is non-zero, though it may be reasonable to hope that the conditional expectation is close to being linear in many cases.

Next, we provide an assumption to eliminate the bias terms discussed above.
\begin{assumption}[Additional Assumptions to Rule Out Bias Terms] \label{ass:bias-elimination} \

    \begin{itemize}
        \item[(A)] The path of untreated potential outcomes does not depend on time-invariant covariates.  That is, $\E[\Delta Y_{t^*}(0) | X_{t^*}, X_{t^*-1}, Z, D=0] = \E[\Delta Y_{t^*}(0) | X_{t^*}, X_{t^*-1}, D=0].$
        \item[(B)] The path of untreated potential outcomes only depends on the change in time-varying covariates.  That is, $\E[\Delta Y_{t^*}(0) | X_{t^*}, X_{t^*-1}, D=0] = \E[\Delta Y_{t^*}(0) | \Delta X_{t^*}, D=0].$
        \item[(C)] The path of untreated potential outcomes is linear in the change in time-varying covariates.  That is, $\E[\Delta Y_{t^*}(0) | \Delta X_{t^*}, D=0] = \L_0(\Delta Y_{t^*}|\Delta X_{t^*}).$
    \end{itemize}
\end{assumption}

\begin{theorem} \label{thm:twfe-attX} Under \Cref{ass:sampling,ass:overlap,ass:conditional-parallel-trends}, and if, in addition,  \Cref{ass:bias-elimination} also holds, then
\begin{align*}
    \alpha = \E\Big[ w(\Delta X_{t^*}) \mathrm{ATT}(X_{t^*},X_{t^*-1},Z) \Big| D=1 \Big],
\end{align*}
where the weights $w(\Delta X_{t^*})$ are the same ones defined in \Cref{thm:twfe}.  If, in addition, $\mathrm{ATT}(X_{t^*},X_{t^*-1},Z)$ is constant across all values of $(X_{t^*},X_{t^*-1},Z)$, then
    \begin{align*}
        \alpha = \mathrm{ATT}.
    \end{align*}
\end{theorem}

\Cref{thm:twfe-attX} provides sufficient conditions for $\alpha$ from \Cref{eqn:fd} to be equal to a weighted average of conditional ATTs under the conditional parallel trends assumption in \Cref{ass:conditional-parallel-trends}.  The intuition for the result is that the conditions in \Cref{ass:bias-elimination} together imply that
\begin{align*}
    \E[\Delta Y_{t^*} | X_{t^*}, X_{t^*-1}, Z, D=0] = \L_0(\Delta Y_{t^*} | \Delta X_{t^*}).
\end{align*}
This is sufficient for the bias terms in \Cref{thm:twfe} to be equal to 0, and, thus, $\alpha$ is equal to a weighted average of $\mathrm{ATT}(X_{t^*},X_{t^*-1},Z)$.

The result in \Cref{thm:twfe-attX} suggests several potential issues with the TWFE regression in \Cref{eqn:fd}.  First, the additional conditions in \Cref{ass:bias-elimination} are likely to be strong in many applications, and, perhaps more importantly, it is very uncommon for empirical work to grapple with whether or not these types of assumptions are plausible in a given application.

Second, even if one is willing to maintain the additional assumptions in \Cref{ass:bias-elimination}, $\alpha$ from the TWFE regression is still hard to interpret for several reasons.  %
The first issue with interpreting $\alpha$ is that, although the weights have mean one, it is possible to have negative weights for some values of $\mathrm{ATT}(X_{t^*},X_{t^*-1},Z)$.  This can happen for values of the covariates among the treated group where $\L(D|\Delta X_{t^*}) > 1$, which is possible because $\L(D|\Delta X_{t^*})$ is a linear projection of a binary treatment on $\Delta X_{t^*}$ that is not restricted to be between 0 and 1.  Negative weights have often been emphasized as being particularly problematic (see, for example, \citet{chaisemartin-dhaultfoeuille-2020,blandhol-bonney-mogstad-torgovitsky-2025}). %
For example, negative weights imply that it is possible to come up with examples where $\mathrm{ATT}(X_{t^*},X_{t^*-1},Z)$ is positive for all values of the covariates, but $\alpha$ could be negative due to the weighting scheme.  In empirical work, estimating $\L(D|\Delta X_{t^*})$ and checking if there are negative weights is straightforward.  Another issue is that the weights have a \textit{weight-reversal} property (we adapt this terminology from \citet{sloczynski-2022}).  Notice that the ideal weighting scheme would be for $w(\Delta X)$ to be uniformly equal to one, in which case, $\alpha=\mathrm{ATT}$.  Relative to this natural baseline, the weights in \Cref{thm:twfe-attX} indicate that $\alpha$ tends to put too much weight on conditional ATTs for values of the covariates that are relatively uncommon among the treated group relative to the untreated group and puts too little weight on conditional ATTs for values of the covariates that are relatively common among the treated group relative to the untreated group.

Finally, if, in addition to all the previous conditions, conditional ATTs are constant across different values of the covariates, then $\alpha = \mathrm{ATT}$.  This is a treatment effect homogeneity condition with respect to the covariates.\footnote{\citet{meyer-1995,abadie-2005,santanna-zhao-2020} all mention that unmodeled treatment effect heterogeneity with respect to the covariates leads $\alpha$ to not be equal to the ATT.  The first term in the expression for $\alpha$ in \Cref{thm:twfe} provides an explicit expression for $\alpha$ when there is treatment effect heterogeneity, and, in agreement with those papers, our result indicates that $\alpha=\mathrm{ATT}$ when there is no treatment effect heterogeneity with respect to the covariates.}   It \textit{is} somewhat weaker than individual-level treatment effect homogeneity, and it allows for treatment effects to be systematically different for treated units relative to untreated units.  Instead, it says that, for the treated group, treatment effects cannot be systematically different across different values of the covariates.  However, this assumption is likely to be very strong in most economic applications, and it is not commonly considered in empirical work.

These results differ greatly from our earlier result on identifying the ATT in \Cref{eqn:att}.  That result did not require any of the additional conditions in \Cref{ass:bias-elimination}.  %

\begin{remark}[Alternative conditions on the propensity score for interpreting $\alpha$] \label{rem:alternative-pscore-conditions}
    One can also show that $\alpha$ is equal to a weighted average of conditional ATTs under restrictions on the propensity score (rather than restrictions on $\E[\Delta Y_t(0) | X_{t^*}, X_{t^*-1}, Z, D=0]$ as above); namely, $\P(D=1| X_{t^*}, X_{t^*-1}, Z) = \L(D|\Delta X_{t^*})$. See \citet{angrist-1998,aronow-samii-2016,sloczynski-2022} for results along these lines with cross-sectional data under unconfoundedness.  In \Cref{sa:pscore-no-bias}, we argue that, in the panel data context that we consider,  linearity conditions are less plausible
 on the propensity score than on the outcome models discussed above.  Moreover, some leading cases where the propensity score would be linear by construction in cross-sectional settings do not apply in our setting.
\end{remark}

\begin{remark}[Comparison to conditions for other estimation strategies] \label{rem:comparison-to-other-estimation-strategies}
    Interestingly, very similar restrictions as the ones discussed in \Cref{ass:bias-elimination} arise in some recently proposed ``heterogeneity robust'' versions of difference-in-differences.  For example, the imputation approaches proposed in \citet{gardner-thakral-to-yap-2023,borusyak-jaravel-spiess-2024} involve estimating the model $Y_{it}(0) = \theta_t + \eta_i + X_{it}'\beta + e_{it}$ (see \citet[Eq.\,(7)]{gardner-thakral-to-yap-2023} and \citet[Eq.\,(5)]{borusyak-jaravel-spiess-2024}) which, in the two-period context considered here, implicitly uses the assumption that $\E[\Delta Y_{t^*}(0) | X_{t^*}, X_{t^*-1}, Z, D=0] = \L_0(\Delta Y_{t^*}|\Delta X_{t^*})$---the same condition as implied by \Cref{ass:bias-elimination}. Alternatively, the regression adjustment version of \citet{callaway-santanna-2021} implicitly uses the assumption that $\E[\Delta Y_{t^*}(0) | X_{t^*}, X_{t^*-1}, Z, D=0] = \L_0(\Delta Y_{t^*}|X_{t^*-1},Z)$---besides linearity, this condition effectively says that the path of untreated potential outcomes does not depend on $X_{t^*}$ once one controls for $X_{t^*-1}$ and $Z$.  The estimators we propose below do not include either of these types of auxiliary assumptions.  See Appendix \ref{sa:comparison-to-other-estimation-strategies} for a more detailed comparison.
\end{remark}

\section{Covariate Balance Diagnostics} \label{sec:diagnostics}

\Cref{thm:twfe} highlights several potential sources of bias from using the TWFE regression in \Cref{eqn:fd}.  In this section, we quantitatively assess \textit{how much} these bias terms matter in practice.  This is not an easy task as the conditional expectations in Terms (A)-(C) of \Cref{thm:twfe} are difficult to estimate without imposing additional functional form assumptions.  The misspecification bias terms in Terms (A)-(C) amount to violations of linearity that come from differences between $\E[\Delta Y_{t^*} | X_{t^*},X_{t^*-1},Z,D=0]$ and $\L_0(\Delta Y_{t^*} | \Delta X_{t^*})$.  Below, we propose a simple approach to assess the sensitivity of $\alpha$ from the TWFE regression to possible violations of this linearity condition based on assessing implicit covariate balance.  The second part of this section considers related diagnostics of augmented inverse propensity score weighting (AIPW) estimators of the ATT along the lines of the alternative estimators we propose later in the paper.

To motivate this section's results, note that if we could find ``balancing weights'' $\vartheta_0(X_{t^*},X_{t^*-1},Z)$ that re-weight the untreated group such that it has the same distribution of $(X_{t^*},X_{t^*-1},Z)$ as the treated group, then it would be the case that
\begin{align*}
    \E[\Delta Y_{t^*}(0) | D=1 ] &=
    \E\Big[\E[\Delta Y_{t^*}(0)|X_{t^*},X_{t^*-1},Z,D=0] \Big| D=1\Big] \\
    &= \E\Big[\vartheta_0(X_{t^*},X_{t^*-1},Z) \E[\Delta Y_{t^*}(0)|X_{t^*},X_{t^*-1},Z,D=0] \Big| D=0\Big] \\
    &= \E[\vartheta_0(X_{t^*},X_{t^*-1},Z) \Delta Y_{t^*}|D=0],
\end{align*}
and, therefore, that we could recover $\mathrm{ATT} = \E[\Delta Y_{t^*}|D=1] - \E[\vartheta_0(X_{t^*},X_{t^*-1},Z) \Delta Y_{t^*}|D=0]$.  In other words, if we could balance the distribution of covariates for the untreated group relative to the treated group, then we could recover the path of untreated potential outcomes for the treated group by looking at the mean path of outcomes for the untreated group after it has been re-weighted to have the same distribution of covariates as the treated group.    These sorts of balancing weights are related to a large number of weighting estimators.  For example, in population, the weights from propensity score re-weighting satisfy this property (\citet{rosenbaum-rubin-1983}).

One important property of balancing weights is that they balance functions of the covariates across groups; i.e., for some function of the covariates $g$,
\begin{align} \label{eqn:cov-balance}
    \E[g(X_{t^*},X_{t^*-1},Z) | D=1] = \E[\vartheta_0(X_{t^*},X_{t^*-1},Z) g(X_{t^*},X_{t^*-1},Z)|D=0].
\end{align}
We show that TWFE and AIPW estimators can be re-expressed as weighting estimators with particular weights.  Following the discussion above, we apply these weights to functions of the covariates to check how well these weights balance the covariates across groups.  If the weights do not balance the covariates well, the corresponding estimator is more sensitive to violations of modeling assumptions for the outcome than if the weights balance the covariates well.  Heuristically, unbalanced covariates that have larger effects on the outcome are more problematic than covariates that have smaller effects on the outcome.  Finally, although we emphasize TWFE and AIPW, the same sorts of covariate balance diagnostics could be applied to any DiD estimator that can be expressed as a weighting estimator (e.g., DiD versions of matching, inverse probability tilting, entropy balancing, etc.).

\subsection{TWFE Diagnostics} \label{sec:twfe-diagnostics}

Returning to $\alpha$ from the TWFE regression, a useful insight is that it can be written as a re-weighting estimator.  To see this, notice that it follows from Frisch-Waugh-Lovell arguments that
\begin{align}
    \alpha = \E\left[ \frac{(D-\L(D|\Delta X_{t^*})) \Delta Y_{t^*}}{\E[(D-\L(D|\Delta X_{t^*}))^2]} \right]. \label{eqn:fwl}
\end{align}
Let $\pi := \P(D=1),$ then it immediately follows from the law of iterated expectations that
\begin{equation*}
    \alpha = \E[w_1(\Delta X_{t^*}) \Delta Y_{t^*}|D=1] - \E[w_0(\Delta X_{t^*}) \Delta Y_{t^*} | D=0],
\end{equation*}
where $\displaystyle
    w_1(\Delta X_{t^*}) = \frac{\pi (1-\L(D|\Delta X_{t^*}))}{\E[(D-\L(D|\Delta X_{t^*}))^2]}$   and $\displaystyle  w_0(\Delta X_{t^*}) = \frac{(1-\pi) \L(D|\Delta X_{t^*})}{\E[(D-\L(D|\Delta X_{t^*}))^2]}.$ We refer to the weights $w_d(\Delta X_{t^*})$ as \textit{implicit regression weights} below.  Notice that these weights are simple to calculate, as the most complicated terms are linear projections.  %
Building on the intuition for weighting estimators discussed earlier in this section, the diagnostics we propose in this section come from applying these weights to functions of the covariates to check how well the weights balance the covariates across groups.  In the context of cross-sectional data under the assumption of unconfoundedness, \citet{aronow-samii-2016,chattopadhyay-zubizarreta-2023} derive related weights and discuss a number of properties of these types of weights.  For our purposes, the most notable property is that these weights will balance (in mean) the covariates that show up in the regression; thus, in our case, they will balance $\Delta X_{t^*}$ across groups.  See \Cref{prop:balancing-weights-explanation} for a more detailed explanation of why this is the case.  Although the weights balance the mean of $\Delta X_{t^*}$, they do not necessarily balance the distribution/means of the levels of time-varying covariates (that is, $X_{t^*}$ or $X_{t^*-1}$) or of time-invariant covariates $Z$.  Thus, our strategy below is to assess the sensitivity of the TWFE regression to violations of linearity by comparing terms such as %
\begin{align*}
    & \E[w_1(\Delta X_{t^*}) X_{t^*}|D=1] & \hspace{-35pt} & \textrm{to }~ \E[w_0(\Delta X_{t^*}) X_{t^*} | D=0], & \\
    & \E[w_1(\Delta X_{t^*}) X_{t^*-1} | D=1] & \hspace{-35pt} &\textrm{to }~ \E[w_0(\Delta X_{t^*}) X_{t^*-1} | D=0], \textrm{ or} & \\
    & \E[w_1(\Delta X_{t^*}) Z |D=1] & \hspace{-35pt} & \textrm{to }~ \E[w_0(\Delta X_{t^*}) Z | D=0]. &
\end{align*}
If these terms are all close to each other, it suggests that the implicit regression weights effectively balance time-invariant covariates and the levels of time-varying covariates between the treated group and the untreated group, and, hence, that $\alpha$ from the TWFE regression is not much affected by hidden linearity bias.  On the other hand, if these terms are not close to each other, it suggests that $\alpha$ from the TWFE regression could be sensitive to violations of linearity.  %

\subsection{AIPW Diagnostics} \label{sec:aipw-diagnostics}

The main class of estimators that we suggest as alternatives to the TWFE regression are augmented inverse propensity score weighting (AIPW) estimators.  These estimators involve estimating both an outcome regression model and a model for the propensity score.  In this section, we introduce the particular AIPW estimands that we consider.  Following a similar motivation as in the previous section for TWFE regressions, we recast our AIPW approach as a weighting estimator.  Then, we can apply these implicit AIPW weights to the covariates or functions of the covariates, allowing us to assess how well this estimation strategy balances covariate distributions for the treated and untreated groups.\footnote{The results in this section build on several recent papers that have shown that ostensible outcome models can often be reinterpreted as weighting estimators; these include \citet{robins-sued-lei-rotnitzky-2007,kline-2011,chattopadhyay-zubizarreta-2023}, particularly \citet{chattopadhyay-zubizarreta-2023} though this paper is in the context of cross-sectional data under unconfoundedness.  See Supplementary Appendix \ref{sa:additional-results-two-periods} for additional discussion.} As a step towards developing an AIPW estimator, it is a straightforward extension of the identification results in \Cref{eqn:att} to show (see, e.g., \citet{robins-rotnitzky-zhao-1994,sloczynski-wooldridge-2018,santanna-zhao-2020}) that%

\vspace{-1.1cm}
{\small
\begin{marginalign} \label{eqn:aipw-att}
    ATT = \E\left[ \Delta Y_{t^*} \smallminus \E[\Delta Y_{t^*} | X_{t^*}, X_{t^*\tightminus 1}, Z, D\smallequals 0 ] \Big| D\smallequals 1 \right] - \E\left[ w_0^{aipw} \big(\Delta Y_{t^*} \smallminus \E[\Delta Y_{t^*} | X_{t^*}, X_{t^*\tightminus 1}, Z, D\smallequals 0]\big) \Big| D\smallequals 0 \right],
\end{marginalign}}

\vspace{-.4cm}
\noindent where\footnote{All of the weights in this section are functions of $(X_{t^*},X_{t^*-1},Z)$, but we omit this dependence to simplify notation.} {\small $
\displaystyle
    w_0^{aipw} := \frac{\varpi_0^{aipw}}{\E[\varpi_0^{aipw}|D=0]},\;$} with {\small $\;\;\displaystyle \varpi_0^{aipw} := \frac{(1-\pi) p(X_{t^*}, X_{t^*-1}, Z)}{\pi \big(1-p(X_{t^*},X_{t^*-1},Z)\big)}.$}

\smallskip
Estimating the ATT based on this expression requires first the estimation of $\E[\Delta Y_{t^*}|X_{t^*},X_{t^*-1},Z,$ $D\smallequals 0]$ and $p(X_{t^*}, X_{t^*-1}, Z)$.  In this section, we specify a linear working model, $\L_0(\Delta Y_{t^*}|X_{t^*},X_{t^*-1},Z)$, for the expectation. Similarly, let $\tilde{p}(X_{t^*},X_{t^*-1},Z)$ denote a working model for $p(X_{t^*},X_{t^*-1},Z)$ (leading choices include a logit or probit model, but there are other possibilities).\footnote{To be clear, the proof of \Cref{prop:aipw-decomp} does not require any substantive restrictions on the model for the propensity score, but it does use linearity of the outcome regression model.  That said, the outcome regression model could include interactions, higher order terms, etc.}  %
We allow for the possibility that either or both of these models are misspecified.  Given these working models for the outcome regression and the propensity score, we define

\vspace{-1.1cm}
{\small \begin{equation}
    \label{eqn:aipw-att-estimator}
    \widetilde{\mathrm{ATT}} = \E\left[ \Delta Y_{t^*} \smallminus \L_0(\Delta Y_{t^*}|X_{t^*},X_{t^*\tightminus 1},Z) \Big| D\smallequals 1 \right] - \E\left[ \tilde{w}_0^{aipw} \big(\Delta Y_{t^*} \smallminus \L_0(\Delta Y_{t^*}|X_{t^*},X_{t^*\tightminus 1},Z)\big) \Big| D\smallequals 0 \right]
\end{equation}}

\vspace{-.4cm}
\noindent where {\small $\displaystyle
    \tilde{w}_0^{aipw} := \frac{\tilde{\varpi}_0^{aipw}}{\E[\tilde{\varpi}_0^{aipw}|D=0]},$} with {\small $\displaystyle\tilde{\varpi}_0^{aipw} := \frac{(1-\pi) \tilde{p}(X_{t^*}, X_{t^*-1}, Z)}{\pi \big(1-\tilde{p}(X_{t^*},X_{t^*-1},Z)\big)}$}.

$\widetilde{\mathrm{ATT}}$ is a parametric AIPW estimand corresponding to the ATT expression in Equation \ref{eqn:aipw-att} \phantom{\small{(0)}} but with working models replacing the outcome regression and propensity score.  The sample analog of $\widetilde{\mathrm{ATT}}$ is doubly robust, in the sense that $\widetilde{\mathrm{ATT}}=\mathrm{ATT}$ if either $\E[\Delta Y_{t^*}|X_{t^*},X_{t^*-1},Z,D=0] = \L_0(\Delta Y_{t^*} | X_{t^*}, X_{t^*-1}, Z)$ or $p(X_{t^*},X_{t^*-1},Z) = \tilde{p}(X_{t^*}, X_{t^*-1},Z)$, (i.e., if either the outcome regression model or the propensity score model is correctly specified).  The following proposition shows that $\widetilde{\mathrm{ATT}}$ can be written as a re-weighting estimator.
\begin{proposition} To conserve on notation, let $X=(X_{t^*}, X_{t^*-1}, Z)$.  Define $\gamma_0$ as the linear projection coefficient from projecting $p(X)/\big(1-p(X)\big)$ on $X$; similarly define $\tilde{\gamma}_0$ as the linear projection coefficient from projecting $\tilde{p}(X)/\big(1-\tilde{p}(X)\big)$ on $X$.  Then, under \Cref{ass:sampling,ass:conditional-parallel-trends,ass:overlap}, \label{prop:aipw-decomp}
    \begin{align*}
        \widetilde{\mathrm{ATT}} = \E\left[ \vartheta_1^{aipw} \Delta Y_{t^*} \Big| D=1 \right] - \E\left[ \vartheta_0^{aipw} \Delta Y_{t^*} \Big| D=0 \right],
    \end{align*}
    where $\vartheta_1^{aipw}$ and $\vartheta_0^{aipw}$ are weights defined as
    \begin{align*}
        \vartheta_1^{aipw} := 1 \qquad \textrm{and} \qquad \vartheta_0^{aipw} := \tilde{w}_0^{aipw} + \frac{\gamma_0' X}{\E[\gamma_0'X|D=0]} - \frac{\tilde{\gamma}_0'X}{\E[\tilde{\gamma}_0'X|D=0]},
    \end{align*}
    such that $\E[\vartheta_1^{aipw} | D=1] = \E[\vartheta_0^{aipw} | D=0] = 1$ and $\E[\vartheta_0^{aipw} X | D=0] = \E[X|D=1]$.
\end{proposition}

The proof of \Cref{prop:aipw-decomp} is provided in Supplementary Appendix \ref{sa:additional-results-two-periods}.
\Cref{prop:aipw-decomp} shows that the parametric AIPW estimand $\widetilde{\mathrm{ATT}}$ can be re-formulated as a weighting estimator.  It is possible for the weights to be negative; in applications, it is straightforward to calculate the sample analog of the weights---see Supplementary Appendix \ref{sa:additional-results-two-periods} for more details. The main takeaway from \Cref{prop:aipw-decomp} is that, unlike the implicit TWFE weights discussed above, the implicit AIPW weights balance the levels of time-varying covariates and time-invariant covariates across groups.

\begin{remark}[Regression adjustment and IPW as special cases of AIPW]
\label{rem:ra-ipw-as-speical-cases-of-aipw}
    Two special cases of the parametric AIPW estimand are worth mentioning.  First, if we set $\tilde{p}(X_{t^*},X_{t^*-1},Z) = \pi$ (i.e., no covariates enter the propensity score working model), the second term in \Cref{eqn:aipw-att-estimator} equals zero and $\widetilde{\mathrm{ATT}}$ reduces to a regression adjustment estimand.  Second, if the outcome regression working model includes only an intercept (i.e., no covariates enter the outcome regression), $\widetilde{\mathrm{ATT}}$ reduces to an inverse propensity score weighting (IPW) estimand.  Our results in this section therefore cover both of these cases as well.
\end{remark}

\begin{remark}[Additional Diagnostics] \label{rem:additional-diagnostics}
    Although we emphasize covariate balance with respect to the first moment of the covariates, once we have re-formulated TWFE and AIPW as weighting estimators, all of the tools in the covariate balance checking toolkit become available, e.g., comparing higher order moments of covariates after re-weighting, plotting the distribution of covariates after re-weighting, and calculating the implied target population and effective sample size of the estimator.  See \citet{austin-stuart-2015,imbens-rubin-2015} for substantially more details.
\end{remark}

\section{Multiple Periods and Variation in Treatment Timing} \label{sec:multiple-periods}

In this section, we extend the two-period analysis to a setting with multiple periods and variation in treatment timing across units.  This setting is common in empirical work in economics and has been studied in several recent papers (\citet{chaisemartin-dhaultfoeuille-2020,goodman-bacon-2021,callaway-santanna-2021,sun-abraham-2021}, among others).  The proofs of all of the results in this section are provided in Supplementary Appendix \ref{sa:multiple-periods}.

To start with, we introduce some additional notation and discuss how to extend the assumptions from \Cref{sec:two-periods-setup} to the setting considered here (we provide formal versions of these assumptions as \Cref{ass:staggered,ass:no-anticipation,ass:mp-sampling,ass:mp-overlap} in Supplementary Appendix \ref{sa:multiple-periods}). Let $\T$ denote the number of time periods.  In this section, we allow for $\T$ to be larger than two, but we focus on ``short'' panels where $\T$ is considered to be fixed.  We consider a setting with staggered treatment adoption, where (i) no units are treated in the first period (or units that are treated in the first period are dropped) and (ii) treatment timing can vary across units, but once a unit becomes treated, it remains treated in subsequent periods.  Under staggered treatment adoption, a unit's entire sequence of treatments is fully characterized by its ``group'' where group refers to the period when the unit became treated.  Let $G_i$ denote a unit's group and denote the full set of groups by $\mathcal{G} \subseteq \{2, \ldots, \T+1\}$.  We use the convention of setting $G_i = \T+1$ among units that do not participate in the treatment in any period from $2, \ldots, \T$,\footnote{In the literature, it is somewhat more common to set $G_i = \infty$ for never-treated units, but setting $G_i = \T+1$ unifies some of the notation for the TWFE decomposition results presented below.} and we define $\bar{\mathcal{G}} := \mathcal{G} \setminus \{\T+1\}$ as the set of groups that participate in the treatment in any period.  It is also convenient to define a binary indicator for the never-treated group: let $U_i=1$ for units that never participate in the treatment and $U_i=0$ otherwise.

Let $Y_{it}$ denote the observed outcome for unit $i$ in time period $t$.  Under staggered treatment adoption, we can define potential outcomes based on a unit's group; that is, let $Y_{it}(g)$ denote the potential outcome for unit $i$ in time period $t$ if it were in group $g$.  In terms of potential outcomes, the observed outcome is $Y_{it} = Y_{it}(G_i)$.  In other words, the observed outcome is the potential outcome according to unit $i$'s actual group.  To make the notation more transparent, we also define $Y_{it}(0)$ to be unit $i$'s potential outcome in time period $t$ if it never participated in the treatment.  We also make a no-anticipation assumption that says that, in periods before a unit is treated, its observed outcomes are untreated potential outcomes.  This rules out that the treatment affects outcomes in periods before the treatment actually occurs.  Next, define $X_{it}$ to be a $k \times 1$ vector of time-varying covariates, and let $\mathbf{X}_i := (X_{i1}', X_{i2}', \ldots, X_{i\T}')'$ denote the $Tk \times 1$ vector that stacks the time-varying covariates across periods.  Finally, we continue to use $Z_i$ to denote an $l \times 1$ vector of time-invariant covariates.  Besides the staggered treatment adoption and no-anticipation assumptions, we also maintain an iid sampling assumption, a multi-period version of overlap, and a multi-period version of conditional parallel trends, the latter of which we provide here:

\begin{namedassumption}{MP-PT}[Multi-Period Parallel Trends] \label{ass:mp-parallel-trends} For $t=2,\ldots,\T$ and for all $g \in \mathcal{G}$,
\begin{align*}
    \E[\Delta Y_t(0) | \mathbf{X},Z,G=g] = \E[\Delta Y_t(0) | \mathbf{X},Z].
\end{align*}
\end{namedassumption}

Following \citet{callaway-santanna-2021} and \citet{wooldridge-2025}, we target the identification of group-time average treatment effects, defined as
\begin{align*}
    \mathrm{ATT}(g,t) := \E[Y_t(g) - Y_t(0) | G=g].
\end{align*}
$\mathrm{ATT}(g,t)$ is the average treatment effect for group $g$ in period $t$.  We also define the conditional-on-covariates version of group-time average treatment effects
\begin{align*}
    \mathrm{ATT}_{g,t}(\mathbf{x},z) := \E[Y_t(g) - Y_t(0) | \mathbf{X}=\mathbf{x}, Z=z, G=g].
\end{align*}
In \Cref{prop:mp-attgt}, we show that both of these are identified and can be expressed as
\begin{align*}
    \mathrm{ATT}_{g,t}(\mathbf{X},Z) = \E[Y_t - Y_{g-1} | \mathbf{X}, Z, G=g] - \E[Y_t - Y_{g-1} | \mathbf{X}, Z, U=1]
\end{align*}
and \vspace{-32.2pt}
\begin{align*}
    \mathrm{ATT}(g,t) &= \E[Y_t - Y_{g-1} | G=g] - \E\Big[ \E[Y_t - Y_{g-1} | \mathbf{X}, Z, U=1] \Big| G=g\Big].
\end{align*}
This result generalizes the identification result in \Cref{eqn:att} from a setting with two time periods to one with staggered treatment adoption.  The argument closely follows the identification result for $\mathrm{ATT}(g,t)$ in \citet{callaway-santanna-2021} except that some covariates can be time-varying.%

Group-time average treatment effects are important building blocks for our results below on interpreting TWFE regressions.  However, unlike $\alpha$ from the TWFE regression in \Cref{eqn:twfe}, they are functional parameters in the sense that they can vary arbitrarily across $g$ and $t$.  Therefore, it is more natural to compare $\alpha$ from the TWFE regression to an aggregated causal effect parameter; in particular, we consider the following overall average treatment effect on the treated parameter
\begin{align*}
    \mathrm{ATT}^o := \E\Big[ \bar{Y}^{post} - \bar{Y}(0)^{post} \Big| U=0\Big],
\end{align*}
where, for units that ever participate in the treatment, we define
\begin{align*}
    \bar{Y}_i^{post} := \frac{1}{\T-G_i+1} \sum_{t=G_i}^{\T} Y_{it} \qquad \textrm{and} \qquad \bar{Y}_i(0)^{post} := \frac{1}{\T-G_i+1} \sum_{t=G_i}^{\T} Y_{it}(0).
\end{align*}
These are the average observed outcome and average untreated potential outcome, respectively, across unit $i$'s post-treatment time periods.  Thus, $\mathrm{ATT}^o$ is the average treatment effect across the population that participates in the treatment in any time period.  \citet{callaway-santanna-2021} show that
\begin{equation*}
    \mathrm{ATT}^o = \sum_{g \in \bar{\mathcal{G}}}   \sum_{t=g}^{\T} w^o(g,t) \mathrm{ATT}(g,t),
\end{equation*}
where $w^o(g,t) := \bar{p}_g/(\T-g+1)$ and $\bar{p}_g := \P(G=g | G \in \bar{\mathcal{G}})$, which is the probability of being in group $g$ conditional on being among the set of groups that ever participates in the treatment.

\subsubsection*{TWFE Decomposition}

Next, we provide a decomposition of $\alpha$ from \Cref{eqn:twfe} under staggered treatment adoption.  The discussion below uses double-demeaned random variables; for example, $\ddot{Y}_{it} := \displaystyle Y_{it} - \bar{Y}_i - \E[Y_t] + \frac{1}{\T} \sum_{s=1}^{\T} \E[Y_s]$.  We focus on estimating $\alpha$ from \Cref{eqn:twfe} by fixed effects estimation.   Thus, after applying the double-demeaning transformation, we use the following estimating equation:
\begin{equation*}
    \ddot{Y}_{it} = \alpha \ddot{D}_{it} + \ddot{X}_{it}'\beta + \ddot{e}_{it}.
\end{equation*}
Before providing our main results, we need to introduce more notation.  First, notice that a unit's group fully determines $\ddot{D}_{it}$; i.e., $\ddot{D}_{it} = h(G_i,t)$ where
\begin{align*}
    h(g,t) := \indicator{t \geq g} - \frac{\T-g+1}{\T} - \E[D_t] + \frac{1}{\T} \sum_{s=1}^{\T} \E[D_s].
\end{align*}
Next, define the population linear projection of $\ddot{D}_{it}$ on $\ddot{X}_{it}$ as
\begin{align*}
    \L(\ddot{D}_t | \ddot{X}_t) = \ddot{X}_t'\E\left[ \frac{1}{\T}\sum_{s=1}^{\T} \ddot{X}_{s} \ddot{X}_{s}'\right]^{-1} \E\left[ \frac{1}{\T} \sum_{s=1}^{\T} \ddot{X}_{s} \ddot{D}_{s} \right] =: \ddot{X}_t' \Gamma,
\end{align*}
and the population linear projection of $(Y_{it} \tightminus Y_{ig-1})$ on $(X_{it} \tightminus X_{ig-1})$ using the never-treated group as
\begin{align*}
    \L_0\Big(Y_t\tightminus Y_{g-1} \Big| X_t \tightminus X_{g-1}\Big) =: \lambda_{0,t,g-1} + \big(X_t\tightminus X_{g-1}\big)' \Lambda_{0,t,g-1},
\end{align*}
where $\lambda_{0,t,g-1}$ is the intercept and $\Lambda_{0,t,g-1}$ is the slope coefficient, both of which can vary by the period $t$ and the base period $(g\tightminus 1)$.  Furthermore, define $\Lambda_0$ as the vector of coefficients from a TWFE regression of $Y_{it}$ on $X_{it}$ using only the never-treated group (see \Cref{eqn:Lambda0} for the complete expression), and define $\lambda_t := \E[Y_t - X_t'\Lambda_0 | U=1]$. Finally, define
\begin{align*}
    \xi_{t,g-1}(\mathbf{X},Z) := \E[Y_t \tightminus Y_{g-1}|\mathbf{X},Z,U \tightequals 1] - \Big((\lambda_t - \lambda_{g-1}) + (X_t \tightminus X_{g-1})' \Lambda_0\Big),
\end{align*}
which will correspond to misspecification bias similar to terms (A), (B), and (C) in \Cref{thm:twfe} (more details below).  Next, we provide our main result relating $\alpha$ to underlying causal effect parameters and misspecification bias terms under staggered treatment adoption.

\begin{theorem} \label{thm:mp-twfe} Under Assumptions \ref{ass:mp-parallel-trends} and \ref{ass:staggered} to \ref{ass:mp-overlap},{\small
    \begin{marginalign} \label{eqn:mp-twfe}
        \alpha = \sum_{g \in \bar{\mathcal{G}}} \sum_{t=g}^{\T} \E\Big[ w^{twfe}_{g,t}(\ddot{X}_t) \Big( \mathrm{ATT}_{g,t}(\mathbf{X},Z) + \xi_{t,g-1}(\mathbf{X},Z)\Big) \Big| G=g \Big] + \sum_{g \in \bar{\mathcal{G}}} \sum_{t=1}^{g-1} \E\Big[ w^{twfe}_{g,t}(\ddot{X}_t) \xi_{t,g-1}(\mathbf{X},Z) \Big| G=g \Big],
    \end{marginalign}}

    \vspace{-.6cm}
    \noindent \hbox{ \footnotesize where $\displaystyle
        w^{twfe}_{g,t}(\ddot{X}_t) :=  \frac{ \Big(h(g,t) - \ddot{X}_t'\Gamma\Big) \pi_g }{\displaystyle \sum_{l \in \bar{\mathcal{G}}} \sum_{s=l}^{\T} \E\Big[\big(h(l,s) - \ddot{X}_{is}'\Gamma\big) \Big| G\tightequals l\Big] \pi_l}$, $ \displaystyle \sum_{g \in \bar{\mathcal{G}}} \sum_{t=g}^{\T} \E\Big[ w_{g,t}^{twfe}(\ddot{X}_t) \Big| G \tightequals g\Big] \tightequals 1$, and $\displaystyle \sum_{g \in \mathcal{G}} \sum_{t=1}^{g-1} \E\Big[ w_{g,t}^{twfe}(\ddot{X}_t) \Big| G \tightequals g\Big] \tightequals \tightminus 1$.}
\end{theorem}

\bigskip
\Cref{thm:mp-twfe} shows that $\alpha$ from the TWFE regression in \Cref{eqn:twfe} is equal to a weighted average of conditional-on-covariates group-time average treatment effects plus misspecification bias terms.  The first term in Equation \ref{eqn:mp-twfe} \phantom{\tiny{(00)}} covers post-treatment periods while the second term covers pre-treatment periods.  The misspecification bias terms arise in both pre- and post-treatment periods; if there are violations of conditional parallel trends, these would also show up in the second term (see \Cref{prop:mp-decomposition3}). This result is analogous to (and extends) the result in \Cref{thm:twfe} in the case with exactly two periods.  Like the earlier case, the weights on conditional group-time average treatment effects are (i) driven by the estimation method, (ii) can be negative, and (iii) sum to one across post-treatment periods.

The following result decomposes the misspecification bias terms in \Cref{thm:mp-twfe}.

\begin{proposition} \label{prop:mp-twfe-misspecification-bias} Under \Cref{ass:staggered,ass:no-anticipation,ass:mp-sampling,ass:mp-overlap} and \ref{ass:mp-parallel-trends}, the misspecification bias terms in \Cref{thm:mp-twfe} can be decomposed as
    \begin{align*}
        \xi_{t,g-1}(\mathbf{X},Z) &= \E[Y_t \tightminus Y_{g-1}|\mathbf{X},Z,U \tightequals 1] - \E[Y_t \tightminus Y_{g-1}|\mathbf{X},U \tightequals 1]  \Big) \tag{MB-1} \\
        & \hspace{15pt} + \Big(\E[Y_t \tightminus Y_{g-1}|\mathbf{X},U \tightequals 1] - \E[Y_t \tightminus Y_{g-1}|X_t, X_{g-1},U \tightequals 1] \Big) \tag{MB-2} \\
        & \hspace{15pt} + \Big(\E[Y_t \tightminus Y_{g-1}|X_t,X_{g-1},U \tightequals 1] - \E[Y_t \tightminus Y_{g-1}|(X_t \tightminus X_{g-1}),U \tightequals 1] \Big) \tag{MB-3} \\
        & \hspace{15pt} + \Big(\E[Y_t \tightminus Y_{g-1}|(X_t \tightminus X_{g-1}),U \tightequals 1] - \big(\lambda_{0,t,g-1} + (X_t \tightminus X_{g-1})'\Lambda_{0,t,g-1} \big) \Big) \tag{MB-4} \\
        & \hspace{15pt} + \Big( \big(\lambda_{0,t,g-1} - (\lambda_t - \lambda_{g-1})\big) + (X_t \tightminus X_{g-1})'(\Lambda_{0,t,g-1} - \Lambda_0) \Big). \tag{MB-5}
    \end{align*}
\end{proposition}
Next, we discuss the components of the misspecification bias terms in \Cref{prop:mp-twfe-misspecification-bias} along with a set of sufficient conditions to eliminate them from the expression for $\alpha$ in \Cref{thm:mp-twfe}.  These conditions rationalize interpreting $\alpha$ from \Cref{eqn:twfe} as a weighted average of $\mathrm{ATT}_{g,t}(\mathbf{X},Z)$.  The conditions are stated formally in \Cref{ass:mp-nobias}.

\newpage
\noindent \paragraph{Conditions to Eliminate Misspecification Bias}
\begin{itemize}[leftmargin=25pt,noitemsep]
    \item [(1)] The path of untreated potential outcomes does not depend on time-invariant covariates.
    \item [(2)] The path of untreated potential outcomes does not depend on time-varying covariates in other periods besides $(g\tightminus 1)$ and $t$.
    \item [(3)] The path of untreated potential outcomes only depends on the change in time-varying covariates between periods $(g \tightminus 1)$ and $t$.
    \item [(4)] The path of untreated potential outcomes is linear in the change in time-varying covariates.
    \item [(5)] The effect of the change in time-varying covariates over time on the path of untreated potential outcomes is constant across time periods.
\end{itemize}

\smallskip
Each condition serves to set the corresponding term in \Cref{prop:mp-twfe-misspecification-bias} to 0.  Conditions (1)-(3) are all required to deal with the multiple-period version of hidden linearity bias: that transforming the model to eliminate the unit fixed effect also changes the functional form of the time-varying covariates and eliminates the time-invariant covariates, and, hence, effectively results in changing the parallel trends assumption.  Condition (4) is a linearity requirement, similar to Condition (C) in \Cref{ass:bias-elimination} in the two-period case.  %
Condition (5) has no immediate analog in the two-period case.  It says that the path of untreated potential outcomes can depend on the magnitude of changes in time-varying covariates, but the impact of a given change should not vary across time periods.

Although these conditions eliminate misspecification bias, we show in the Supplementary Appendix (\Cref{thm:mp-twfe-nobias}) that, similarly to the two period case, even if these conditions hold, $\alpha$ is still equal to a weighted average of $\mathrm{ATT}_{g,t}(\mathbf{X},Z)$ with weights that are (i) non-transparently driven by the estimation method, (ii) difficult to rationalize, and (iii) can be negative (as pointed out in \citet{chaisemartin-dhaultfoeuille-2020}). We also show that, under additional restrictions on treatment effect heterogeneity with respect to covariates, groups, and time periods, $\alpha=\mathrm{ATT}$. However, like the case with two periods, these extra conditions are likely to be very strong in most applications.  The results in Theorems \ref{thm:mp-twfe} and \ref{thm:mp-twfe-nobias} are related to results in several other papers, including \citet{goodman-bacon-2021,lin-zhang-2022,ishimaru-2022}.  See Supplementary Appendix \ref{sa:multiple-periods} for a detailed discussion.

\subsubsection*{AIPW Estimands with Multiple Periods}

Next, we consider parametric AIPW estimands for group-time average treatment effects and the overall average treatment effect with multiple periods and variation in treatment timing.  This is the population version of our main alternative estimator to the TWFE regression; see the next section for further details.  Define the parametric AIPW estimand for $\mathrm{ATT}(g,t)$ as
\begin{align} \label{eqn:mp-aipw-estimand}
    \widetilde{ATT}^{aipw}(g,t) &= \E\Big[ (Y_t \tightminus Y_{g-1}) \smallminus \tilde{\L}^0_{g,t}(Y_t\tightminus Y_{g-1}|\mathbf{X},Z) \Big| G\smallequals g \Big] \nonumber \\
    & - \E\Big[ \tilde{w}^{0,aipw}_{g,t}(\mathbf{X},Z) \big((Y_t \tightminus Y_{g-1}) \smallminus \tilde{\L}^0_{g,t}(Y_t\tightminus Y_{g-1}|\mathbf{X},Z)\big) \Big| U\smallequals 1 \Big],
\end{align}
where $\tilde{\L}^0_{g,t}(Y_t\tightminus Y_{g-1}|\mathbf{X},Z)$ is the linear projection of $Y_t\tightminus Y_{g-1}$ onto $\mathbf{W}_{g,t}$ among comparison units, with $\mathbf{W}_{g,t}$ a possibly $(g,t)$-varying vector of regressors that may use covariate values from selected time periods (e.g., $X_t$, $X_{g-1}$, and $Z$) and may include interactions and higher-order terms, and {\small $\displaystyle
    \tilde{w}^{0,aipw}_{g,t} := \frac{\tilde{\varpi}^{0,aipw}_{g,t}(\mathbf{X},Z)}{\E[\tilde{\varpi}^{0,aipw}_{g,t}(\mathbf{X},Z)|U=1]},$} with {\small $\displaystyle  \tilde{\varpi}^{0,aipw}_{g,t}(\mathbf{X},Z) = \frac{\pi_0 \tilde{p}_{g,t}(\mathbf{X},Z)}{\pi_g \big(1-\tilde{p}_{g,t}(\mathbf{X},Z)\big)}$,}  $\pi_g := \P(G=g)$, $\pi_0 := \P(U=1)$, and $\tilde{p}_{g,t}(\mathbf{X},Z)$ is a parametric working model for the generalized propensity score
\begin{align*}
    p_{g}(\mathbf{X},Z) := \P(G=g | \mathbf{X},Z,\indicator{G=g} + U = 1),
\end{align*}
which is the conditional probability of being in group $g$ conditional on being in group $g$ or the never-treated group.  We index $\tilde{p}_{g,t}(\mathbf{X},Z)$ by $g$ and $t$ to allow the working model for the generalized propensity score to change across time periods, particularly with respect to which time-varying covariates are included in the model.  Note that $\widetilde{\mathrm{ATT}}^{aipw}(g,t)$ is a model-dependent quantity that depends on the choice of working models $\tilde{\L}^0_{g,t}$ and $\tilde{p}_{g,t}$, and therefore differs in general from the causal target parameter $\mathrm{ATT}(g,t)$.  It is straightforward to show, analogously to the two-period case, that the sample analog of $\widetilde{\mathrm{ATT}}^{aipw}(g,t)$ is doubly robust for $\mathrm{ATT}(g,t)$ (see the next section for more details).  In addition, define the following parametric AIPW estimand for $\mathrm{ATT}^o$:
\begin{align} \label{eqn:mp-aipw-estimand-o}
    \widetilde{\mathrm{ATT}}^{aipw,o} = \sum_{g \in \bar{\mathcal{G}}} \sum_{t=g}^{\T} \widetilde{\mathrm{ATT}}^{aipw}(g,t) w^o(g,t).
\end{align}

\subsubsection*{Covariate Balance Diagnostics}

In Supplementary Appendix \ref{sa:multiple-periods}, we extend the covariate balance diagnostics for TWFE and AIPW with two periods (\Cref{sec:diagnostics}) to the staggered treatment adoption setting considered in this section.  In particular, in \Cref{prop:mp-alpha-balancing-weights-gmin1}, we show that $\alpha$ can be rewritten in terms of implicit regression weights:

\vspace{-1.1cm}
{\small
\begin{align} \label{eqn:alpha-mp-diagnostics}
    \alpha &= \sum_{g \in \bar{\mathcal{G}}} \sum_{t=g}^{\T} \bar{w}^{twfe}(g,t) \left\{ \E\left[ w^{1,twfe}_{g,t}(\mathbf{X},Z) (Y_t\smallminus Y_{g\tightminus 1}) \Big| G\smallequals g\right] - \E\left[ w^{0,twfe}_{g,t}(\mathbf{X},Z) (Y_t\smallminus Y_{g\tightminus 1}) \Big| U\smallequals 1\right] \right\} \\
    & + \sum_{g \in \bar{\mathcal{G}}} \sum_{t=1}^{g-1} \bar{w}^{twfe}(g,t) \left\{ \E\left[ w^{1,twfe}_{g,t}(\mathbf{X},Z) (Y_t\smallminus Y_{g\tightminus 1}) \Big| G\smallequals g\right] - \E\left[ w^{0,twfe}_{g,t}(\mathbf{X},Z) (Y_t\smallminus Y_{g\tightminus 1}) \Big| U\smallequals 1\right] \right\} + r, \nonumber
\end{align}
}

\vspace{-.4cm}
\noindent where the sum weights {\small $\displaystyle \bar{w}^{twfe}(g,t) := \E[w^{twfe}_{g,t}(\ddot{X}_t) | G=g]$}, the expectation weights  {\small $\displaystyle
    w^{1,twfe}_{g,t}(\mathbf{X},Z) := \frac{(\ddot{D}_t - \ddot{X}_t'\Gamma)}{\E[(\ddot{D}_t - \ddot{X}_t'\Gamma) | G=g] }$} and {\small $\displaystyle w^{0,twfe}_{g,t}(\mathbf{X},Z) := \frac{(\ddot{D}_t - \ddot{X}_t'\Gamma)}{\E[(\ddot{D}_t - \ddot{X}_t'\Gamma)|U=1]}$},
and $r$ is a remainder term.\footnote{The remainder term is a byproduct of using $(g-1)$ as the base period in the decomposition of $\alpha$ presented here.  In the discussion after \Cref{prop:mp-alpha-balancing-weights-gmin1}, we argue that this term is likely to be small in most applications, and, indeed, this term is negligible in all diagnostics we report in our application.  Very similar arguments can rationalize a different base period choice, such as $(g-2)$ or $(g-3)$, which could be attractive in applications with anticipation effects (see \citet{callaway-santanna-2021,sun-abraham-2021}).  In \Cref{prop:mp-alpha-balancing-weights3}, we also provide a decomposition that uses period one as the base period, and it involves exactly the same weights but does not include a remainder term.  We prefer the decomposition using $(g-1)$ as the base period because (i) it allows a direct comparison with the parametric AIPW estimand discussed below, where differences are due entirely to differences in the implicit weighting schemes, and (ii) it allows us to quantify how much pre-treatment violations of parallel trends contribute to $\alpha$.\label{fn:mp-twfe-decomp-base-period}}

Similarly, we show in \Cref{prop:mp-aipw-diagnostics} that $\widetilde{\mathrm{ATT}}^{aipw,o}$ can be rewritten in terms of weighted averages of paths of outcomes for each group relative to the never-treated group.  In particular,

\vspace{-1.1cm}
{ \small
\begin{marginalign}\label{eqn:aipw-mp-diagnostics}
    \widetilde{ATT}^{aipw,o} = \sum_{g \in \bar{\mathcal{G}}} \sum_{t=g}^{\T} w^o(g,t) \left\{ \E\left[ \vartheta^{1,aipw}_{g,t}(\mathbf{X},Z) (Y_t\smallminus Y_{g\tightminus 1}) \Big| G\smallequals g\right] - \E\left[ \vartheta^{0,aipw}_{g,t}(\mathbf{X},Z) (Y_t\smallminus Y_{g\tightminus 1}) \Big| U\smallequals 1\right] \right\},
\end{marginalign}
}

\vspace{-.4cm}
\noindent where $w^o(g,t)$ are the same weights as in $\mathrm{ATT}^o$ above, {\small $\vartheta^{1,aipw}_{g,t}(\mathbf{X},Z) := 1 $}, and

\vspace{-1.1cm}
{\small \begin{align*}
    \vartheta^{0,aipw}_{g,t}(\mathbf{X},Z) := \tilde{w}^{0,aipw}_{g,t}(\mathbf{X},Z) + \frac{\gamma_{g,t}'\mathbf{W}_{g,t}}{\E[\gamma_{g,t}'\mathbf{W}_{g,t}|U=1]} - \frac{\tilde{\gamma}_{g,t}'\mathbf{W}_{g,t}}{\E[\tilde{\gamma}_{g,t}'\mathbf{W}_{g,t}|U=1]}.
\end{align*}}

\vspace{-.4cm}
\noindent where $\mathbf{W}_{g,t}$ denotes the covariates used in $\tilde{\L}^0_{g,t}(Y_t\tightminus Y_{g-1}|\mathbf{X},Z)$, $\gamma_{g,t}$ is the linear projection coefficient from projecting $p_g(\mathbf{X},Z)/(1-p_g(\mathbf{X},Z))$ on $\mathbf{W}_{g,t}$, and $\tilde{\gamma}_{g,t}$ is the linear projection coefficient from projecting $\tilde{p}_{g,t}(\mathbf{X},Z)/(1-\tilde{p}_{g,t}(\mathbf{X},Z))$ on $\mathbf{W}_{g,t}$.

Applying the TWFE or AIPW weights to (functionals of) $\mathbf{X}$ and $Z$ allows the researcher to assess how well each approach implicitly balances covariates.  The key differences between TWFE and AIPW covariate balance properties are: (i) TWFE weights depend only on transformed time-varying covariates, not levels or time-invariant covariates, while AIPW balances whatever covariates are included in the model; (ii) AIPW weights balance towards group g (the correct target for $\mathrm{ATT}(g,t)$), while TWFE re-weights both groups; (iii) TWFE can be affected by pre-treatment violations of parallel trends, while AIPW is not; and (iv) both can have negative weights. %

\section{Alternative Estimation Strategies} \label{sec:estimation}

This section discusses alternative estimation strategies that avoid the limitations of TWFE regressions discussed above.  First, we discuss AIPW estimators, which are attractive and straightforward to adapt to our context.  Second, from an empirical perspective, the main complication is that, with panel data and time-varying covariates, the dimension of the covariates can be very large.  We discuss several different dimension reduction ideas in the second part of this section.

To start, define $m_{g,t}(\mathbf{X},Z) := \E[Y_t(0) - Y_{g-1}(0) | \mathbf{X}, Z, U=1]$.  We refer to $m_{g,t}(\mathbf{X},Z)$ as an outcome regression model.  Let $\hat{m}_{g,t}(\mathbf{X},Z)$ and $\hat{p}_{g,t}(\mathbf{X},Z)$ denote estimators of $m_{g,t}(\mathbf{X},Z)$ and the generalized propensity score $p_{g}(\mathbf{X},Z)$, respectively (in line with the discussion above, we allow the model for the generalized propensity score to change across time periods).  %
Then, we consider AIPW estimators of $\mathrm{ATT}(g,t)$ of the form
\begin{align*}
    \widehat{\mathrm{ATT}}^{aipw}(g,t) := \frac{1}{n} \sum_{i=1}^n \Big(\hat{w}^{1,aipw}_{g,t}(\mathbf{X}_i,Z_i) - \hat{w}^{0,aipw}_{g,t}(\mathbf{X}_i,Z_i)\Big) \Big( (Y_t - Y_{g-1}) - \hat{m}_{g,t}(\mathbf{X}_i, Z_i)\Big),
\end{align*}
which, after slightly re-arranging terms, is the sample analog of \Cref{eqn:mp-aipw-estimand} (where here we also allow for the possibility of a nonlinear model for the outcome regression), and where
\begin{align*}
    \hat{w}^{1,aipw}_{g,t}(\mathbf{X}_i,Z_i) :=  \frac{\indicator{G_i=g}}{\hat{\pi}_g} \quad \textrm{and} \quad \hat{w}^{0,aipw}_{g,t}(\mathbf{X}_i,Z_i) := \frac{\indicator{U_i=1} \frac{\hat{p}_{g,t}(\mathbf{X}_i,Z_i)}{1-\hat{p}_{g,t}(\mathbf{X}_i,Z_i)}}{\frac{1}{n} \sum_{j=1}^n \indicator{U_j=1} \frac{\hat{p}_{g,t}(\mathbf{X}_j,Z_j)}{1-\hat{p}_{g,t}(\mathbf{X}_j,Z_j)}}.
\end{align*}
AIPW estimators have been well-studied and have several known properties, as we discuss next.

\subsubsection*{Remarks on AIPW Estimation}

\begin{enumerate}[leftmargin=*, labelsep=5pt]
    \item If we specify parametric models for $m_{g,t}(\mathbf{X},Z)$ and $p_g(\mathbf{X},Z)$%
    , then $\widehat{\mathrm{ATT}}^{aipw}\!\!(g,t)$ is doubly robust for $\mathrm{ATT}(g,t)$. Thus, if either the outcome regression or the propensity score model is correctly specified, then the estimator is consistent for $\mathrm{ATT}(g,t)$. See \citet{robins-rotnitzky-zhao-1994,scharfstein-rotnitzky-robins-1999,sloczynski-wooldridge-2018} for general results on the double robustness property of AIPW estimators and \citet{santanna-zhao-2020} for the specific case of DiD.
    \item Given parametric models for $m_{g,t}(\mathbf{X},Z)$ and $p_g(\mathbf{X},Z)$, asymptotic normality of $\widehat{\mathrm{ATT}}^{aipw}(g,t)$ holds under Assumptions \ref{ass:mp-parallel-trends} and \ref{ass:staggered} to \ref{ass:mp-overlap} and weak regularity conditions following the same arguments as in \citet{callaway-santanna-2021}.
    \item The estimator $\widehat{\mathrm{ATT}}^{aipw}(g,t)$ can be used to construct an estimator of the overall average treatment effect, $\mathrm{ATT}^o$, by averaging over all groups and time periods.  In particular,
    \begin{align*}
        \widehat{\mathrm{ATT}}^o = \sum_{g \in \bar{\mathcal{G}}} \sum_{t=g}^{\T} \hat{w}^o(g,t)\widehat{\mathrm{ATT}}^{aipw}(g,t),
    \end{align*}
    where $\hat{w}^o(g,t) = \frac{\hat{\P}(G=g|U=0)}{T-g+1}$.  This estimator is consistent for $\mathrm{ATT}^o$ and is asymptotically normal under the same conditions discussed above.  Similar results hold for event studies or other aggregated parameters that can be expressed as weighted averages of $\mathrm{ATT}(g,t)$.  These results follow directly from those provided in \citet{callaway-santanna-2021}.%
    \item Regression adjustment is a special case of AIPW when no covariates are included in the generalized propensity score model.  In this case $\widehat{\mathrm{ATT}}^{aipw}(g,t)$ simplifies to
    \begin{align*}
        \widehat{\mathrm{ATT}}^{ra}(g,t) = \frac{1}{n} \sum_{i=1}^n \frac{\indicator{G_i=g}}{\hat{\pi}_g} \Big( Y_{it} - Y_{ig-1} - \hat{m}_{g,t}(\mathbf{X}_i,Z_i) \Big),
    \end{align*}
    where $\widehat{\mathrm{ATT}}^{ra}(g,t)$ denotes the regression adjustment estimator of $\mathrm{ATT}(g,t)$.  In this case, consistent and asymptotically normal estimation of $\mathrm{ATT}(g,t)$ hinges on correct specification of the outcome regression $m_{g,t}(\mathbf{X},Z)$.  Regression adjustment is an important special case because small groups (in terms of the number of observations) are fairly common in empirical work. In such cases, estimating the generalized propensity score can be highly unstable, making regression adjustment a possibly more appropriate alternative.  %
    \item Several extensions to the AIPW estimator above apply in our case, such as accounting for anticipation effects and using alternative comparison groups (e.g., not-yet-treated versus never-treated); see \citet{callaway-santanna-2021} for details. Since these issues are standard, we do not elaborate, though they are often important in empirical work and are supported in our code. %
    \item In cases where the researcher does not wish to specify parametric models for $m_{g,t}(\mathbf{X},Z)$ and $p_g(\mathbf{X},Z)$, essentially the same estimator can be used but with machine learners or nonparametric estimators replacing the parametric models (following similar arguments as those in \citet{chernozhukov-etal-2018}).
\end{enumerate}

\paragraph{Dimension Reduction}\mbox{}

An important practical challenge is that, by construction, the dimension of the covariates in $m_{g,t}(\mathbf{X},Z)$ and $p_g(\mathbf{X},Z)$ is likely to be high as $\mathbf{X}$ is of dimension $T k,$ where $k$ is the number of time-varying covariates.  In most applications, reducing the dimension of the covariates will be desirable.\footnote{\label{fn:marginal-structural-models}Related dimension reduction assumptions also appear in the literature on marginal structural models (e.g., \citet{robins-hernan-brumback-2000}), where outcome models and propensity scores are typically specified as parsimonious functions of covariate history rather than the full covariate trajectory.}  One leading dimension-reducing assumption is that
\begin{equation*}
    m_{g,t}(\mathbf{X},Z) = m_{g,t}( X_t \smallminus X_{g-1}, X_{g-1}, Z) \quad \textrm{and} \quad p_g(\mathbf{X},Z) = p_{g,t}( X_t \smallminus X_{g-1}, X_{g-1}, Z),
\end{equation*}
which says that, in terms of time-varying covariates, the outcome regressions and generalized propensity scores only depend on (i) the change in the time-varying covariates from the base period to the current period and (ii) the level of the time-varying covariates in the base period, rather than the covariates across all time periods.  This type of specification includes both types of covariates that show up in \citet{callaway-santanna-2021} and in imputation approaches such as \citet{gardner-thakral-to-yap-2023,borusyak-jaravel-spiess-2024}.  Other options are to (i) assume that $m_{g,t}(\mathbf{X}, Z) = m_{g,t}(\bar{X},Z)$ and $p_g(\mathbf{X},Z) = p_g(\bar{X},Z)$ where $\bar{X}$ is the average of each time-varying covariate, or (ii) choose the covariates in the outcome regression and generalized propensity score in a data-driven way (see Supplementary Appendix \ref{sa:application} for an example).  Rather than necessarily advocating a particular approach to dimension reduction, we instead emphasize that any approach to dimension reduction should be a carefully and transparently considered step of the analysis rather than being inherited from the estimation strategy, as is the case with TWFE regressions.

\section{Application} \label{sec:application}

In this section, we apply our methods to study the effect of stand-your-ground laws on homicides, building on \citet{cheng-hoekstra-2013}. Stand-your-ground laws remove the duty to retreat in violent altercations. \citet{cheng-hoekstra-2013} study 2000-2010, when 20 states implemented such laws in a staggered fashion. There are contrasting theoretical implications of the effect of stand-your-ground policies on homicides: possible deterrence effects (fewer altercations, fewer homicides) versus increased deadliness (more homicides per altercation).
\citet{cheng-hoekstra-2013} find that stand-your-ground laws increased homicides.%

{ \setlength{\tabcolsep}{9pt}
\renewcommand{\arraystretch}{1.0}

\begin{table}[t!]
    \footnotesize
    \caption{Summary Statistics}
    \label{tab:ss}
    \centering
    \begin{tabular}[t]{lrrrrrrrr}
    \toprule
     & \multicolumn{4}{c}{\textbf{Levels (2000)}} & \multicolumn{4}{c}{\textbf{Changes (2010-2000)}} \\
     & Tr. & Untr. & Diff. & Std.\,$\Delta$ & Tr. & Untr. & Diff. & Std.\,$\Delta$ \\
    \cmidrule(lr){2-5} \cmidrule(lr){6-9}
    \multicolumn{9}{l}{\textbf{Outcome}} \\
    \hspace{10pt}log homicides & 5.19 & 4.70 & 0.49 & 0.34 & 0.08 & -0.03 & 0.11 & 0.34\\
    \addlinespace[0pt]
    \multicolumn{9}{l}{\textbf{Time-Invariant Covariates}} \\
    \hspace{10pt}Midwest & 0.33 & 0.17 & 0.16 & 0.38 & & & & \\
    \hspace{10pt}Northeast & 0.00 & 0.31 & -0.31 & -0.88 & & & & \\
    \hspace{10pt}South & 0.52 & 0.17 & 0.35 & 0.81 & & & & \\
    \hspace{10pt}West & 0.14 & 0.35 & -0.20 & -0.47 & & & & \\
    \addlinespace[0pt]
    \multicolumn{9}{l}{\textbf{Time-Varying Covariates}} \\
    \hspace{10pt}log population & 15.11 & 14.97 & 0.14 & 0.14 & 0.11 & 0.12 & -0.01 & -0.12\\
    \hspace{10pt}log police & 5.75 & 5.71 & 0.03 & 0.17 & -0.01 & 0.00 & -0.01 & -0.07\\
    \hspace{10pt}log prisoners & 6.14 & 5.82 & 0.32 & 0.78 & 0.10 & 0.04 & 0.06 & 0.43\\
    \hspace{10pt}log welfare exp.\ per capita & 6.84 & 6.97 & -0.13 & -0.44 & 0.40 & 0.36 & 0.04 & 0.20\\
    \hspace{10pt}log subsidies per capita & 4.58 & 4.69 & -0.11 & -0.21 & 0.27 & 0.17 & 0.10 & 0.23\\
    \hspace{10pt}log median income & 10.78 & 10.93 & -0.15 & -1.08 & -0.08 & -0.04 & -0.04 & -0.43\\
    \hspace{10pt}poverty rate & 12.59 & 9.99 & 2.60 & 1.06 & 2.68 & 2.11 & 0.57 & 0.46\\
    \hspace{10pt}unemployment rate & 4.05 & 3.69 & 0.36 & 0.42 & 4.93 & 4.86 & 0.08 & 0.04\\
    \hspace{10pt}\% black males 15-24 & 1.99 & 2.83 & -0.85 & -0.20 & 0.15 & 0.18 & -0.03 & -0.07\\
    \hspace{10pt}\% black males 25-44 & 3.51 & 5.31 & -1.80 & -0.22 & -0.27 & -0.66 & 0.39 & 0.34\\
    \hspace{10pt}\% white males 15-24 & 11.31 & 10.43 & 0.88 & 0.05 & -0.36 & 0.27 & -0.63 & -0.34\\
    \hspace{10pt}\% white males 25-44 & 24.52 & 24.58 & -0.06 & -0.00 & -3.59 & -3.90 & 0.31 & 0.04\\
    \bottomrule
    \end{tabular}
    \begin{justify}
        \parbox{\textwidth}{ \normalsize \fontsize{12}{24}\selectfont \textit{Notes:} The table provides summary statistics for the outcomes, time-invariant covariates, and levels and changes in time-varying covariates. States are classified as being treated or untreated based on their treatment status in 2010. The column `Diff.' reports the difference between the average of each variable for the treated group relative to the untreated group. The column `Std.\,$\Delta$.' reports the standardized difference of each variable for the treated group relative to the untreated group, which is the difference divided by the pooled standard deviation.}
    \end{justify}
\end{table}
}

Below, we provide two sets of results using two different subsets of the data.  First, we use a subset of the data that only includes the years 2000 and 2010.  This first dataset is in line with our arguments above for the case of exactly two periods. While we report estimates of the effects of stand-your-ground policies on homicides, much of our main interest is in how different estimation strategies (based on \textit{the same identification strategies}) balance the distribution of covariates. We are able to assess this using the covariate balance diagnostics that we proposed for TWFE and AIPW earlier in the paper.  \Cref{tab:ss} provides summary statistics using the two-period subset of the data and for the full set of covariates used in \citet{cheng-hoekstra-2013}.  There are notable differences in covariates between treated and untreated states. Treated states were more likely to be Southern/Midwestern, had lower median income, more prisoners, and higher population, poverty, and unemployment rates. Some differences also appear in covariate changes: poverty and prisoners increased more in treated states, while median income decreased more.
For our second set of results, we mimic \citet{cheng-hoekstra-2013}'s setting much more closely: we use the full data of all 50 states across all available years, the same set of covariates as in one of their main specifications, and the same sampling weights.

Many of our results in this section are reported in figures that summarize covariate balance after applying the implicit weighting schemes for different estimation strategies.  The figures report the standardized difference between the treated and comparison group for each covariate considered. The standardized difference is the difference between the average value of the covariate for the treated group relative to the untreated group, scaled by the pooled standard deviation of the covariate.  We report both raw covariate balance and covariate balance after applying the implicit TWFE or AIPW weights.\footnote{The figures assess covariate balance, but not whether covariates have the same distribution as for the treated group; the latter holds by construction for AIPW and regression adjustment estimators but not for TWFE.  See \Cref{rem:cov-balance-profile-twfe} for more details.}  %
To give a sense of magnitude: typically, standardized differences of around 0.1 or smaller are considered small; around 0.3 are considered medium; and around 0.5 or larger are considered large. See, for example, \citet{imbens-rubin-2015} for a textbook discussion.

\subsubsection*{Results with Two Periods and only Population and Region as Covariates}

For the first set of results, we consider a highly simplified setting.  The outcome is the log of the number of homicides in a state.  We consider one time-varying covariate: the log of a state's population, and one time-invariant covariate: the Census region (Midwest, Northeast, South, or West). %
The intuition for this identification strategy is that a researcher would like to estimate the impact of the policy by comparing the change in log homicides among treated and untreated states that have similar populations and are located in the same region of the country (although region is not included in the regression, the argument would be that it is implicitly controlled for with the unit fixed effect).  We also assess balance of $\indicator{\log(\text{population}) \leq 15}$ to see how well different approaches balance the fraction of small states, as none of the approaches considered below mechanically force this variable to be balanced.%

\begin{figure}[th]
    \centering
    \caption{Two Period Covariate Balance using TWFE and AIPW}
    \label{fig:two-period-covariate-balance-main-results}
    \footnotesize
    \begin{subfigure}[b]{0.3\textwidth}
        \centering
        $\hat{\alpha} = 0.115 \ (0.100)$ \smallskip

        \includegraphics[width=\textwidth]{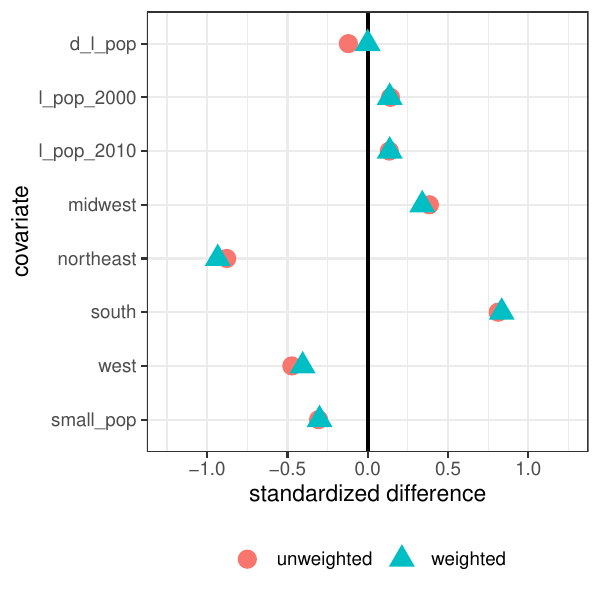}
        \vspace{-20pt}
        \caption{TWFE}
        \label{fig:two-period-covariate-balance-main-results-a}
    \end{subfigure} \hspace{60pt}
    \begin{subfigure}[b]{0.3\textwidth}
        \centering
        $\widehat{\mathrm{ATT}}=0.157 \ (0.107)$ \smallskip

        \includegraphics[width=\textwidth]{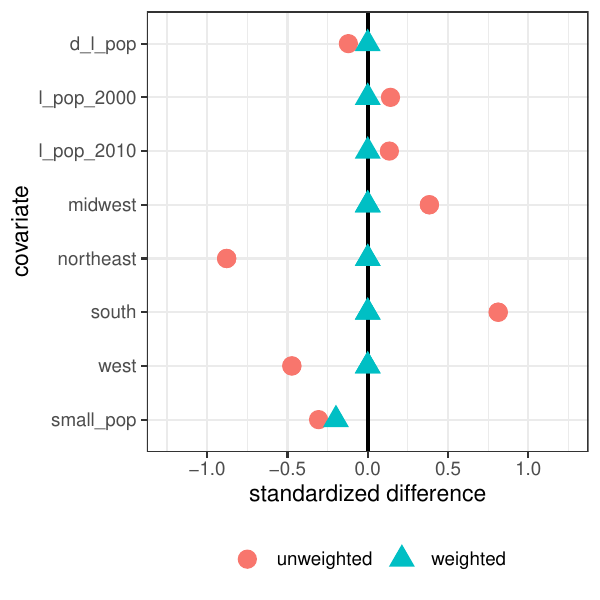}
        \vspace{-20pt}
        \caption{AIPW}
        \label{fig:two-period-covariate-balance-main-results-c}
    \end{subfigure}
    \begin{justify}
        \parbox{\textwidth}{\normalsize \fontsize{12}{24}\selectfont \textit{Notes:} The figure reports estimates of the effects of stand-your-ground laws on homicides and covariate balance statistics using the two-period data discussed in the main text.  The balance statistics are invariant to the outcome.  Different covariates are displayed along the y-axis.  \texttt{d\_l\_pop} is the change in the log of state-level population from 2000 to 2010; \texttt{l\_pop\_2000} and \texttt{l\_pop\_2010} are the level of the log of state-level population in 2000 and 2010, respectively; \texttt{small\_pop} is an indicator for $\log(\text{population}) < 15$; and \texttt{midwest}, \texttt{northeast}, \texttt{south}, \texttt{west} are indicators of Census region.  The x-axis reports standardized differences for the mean of each covariate between the treated group and the untreated group.  The red circles provide the standardized difference for the raw difference, and the blue triangles show the standardized difference after applying the implicit weighting scheme from each estimation method.  Panel (a) comes from regressing $\Delta Y_{t^*}$ on $D_{t^*}$ and $\Delta X_{t^*}$. Panel (b) uses the AIPW estimation strategy discussed in the paper that includes $\Delta X_{t^*}, X_{t^*-1}$, and $Z$ in both the outcome regression model and the propensity score.}
    \end{justify}
\end{figure}

\Cref{fig:two-period-covariate-balance-main-results} summarizes treatment effects and covariate balance. Both panels show qualitatively similar point estimates (TWFE: 0.115, AIPW: 0.157). The AIPW estimate is roughly 40\% larger in magnitude, though neither estimate is statistically different from zero, nor are the estimates statistically different from each other. Because we are only using two time periods, these results are unsurprising.

The covariate balance results are more interesting. In the raw data (the red circles in the figure), there is a small imbalance in population changes, a moderate imbalance in population levels, and a large imbalance in regional distribution. TWFE (Panel a) balances changes in log population (which aligns with the theoretical properties of TWFE), but, in terms of balancing the levels of log population, the regression essentially has no effect: the standardized difference is only 3\% smaller using the implicit TWFE regression weights than in the raw data.  In other words, \textit{controlling for the change in the log of population in the TWFE regression does not result in any more similar treated and untreated groups in terms of log population levels}, which was one of the main goals of including the state's population in the model to begin with.  Similarly, the TWFE regression essentially does not affect the balance of the region indicators or the fraction of small states.  In contrast, AIPW (Panel b) perfectly balances all the covariates included in the model, which is in line with its theoretical properties, and improves balance with respect to the fraction of small states.

Given these results, an interesting follow-up question is: what is the main driver of the improved covariate balance across estimators?  We investigate this question in Supplementary Appendix \ref{sa:application}, where we provide covariate balance statistics for eight different specifications varying both the estimator (between regression adjustment and AIPW) and the covariates included (across (i) changes in time varying covariates only, (ii) pre-treatment levels of time-varying covariates only, (iii) both time-varying and pre-treatment levels of covariates, and (iv) pre-treatment levels of covariates and time-invariant covariates).  The two main takeaways are, first, that the choice of estimator (regression adjustment vs.~AIPW) matters much less than the covariate specification. In fact, both of these give similar estimates to TWFE and have similar covariate balancing properties as TWFE when only the change in time-varying covariates is included.  Second, specifications that directly include pre-treatment level of log population and region indicators do well at balancing the level of log population in each period, suggesting large gains from including the level of a time-varying covariate in at least one period.

\subsubsection*{Results with More Periods and More Covariates}

Next, we use a much closer specification to the one used in \citet{cheng-hoekstra-2013}.  We take the log of homicides per 100,000 people in the state as the outcome, and use sampling weights based on the state's average population.\footnote{See \Cref{rem:sampling-weights} for a discussion on extending our results to include sampling weights.}  We also include additional covariates: log of the number of police per 100,000 population, log of the number of incarcerated persons per 100,000, log of government spending on assistance and subsidies per capita, log of government spending on public welfare per capita, median household income, poverty rate, unemployment rate, and demographic shares of black and white males ages 15–24 and 25–44. %
Because our unit of observation is the state, the size of many of our groups is very small.  This results in AIPW estimation being infeasible (as it is impossible to estimate a generalized propensity score). Therefore, we only report TWFE estimates and regression adjustment estimates.  Finally, in line with \citet{cheng-hoekstra-2013} (but unlike the results above), all the estimates in this section include region-by-year fixed effects.

\begin{figure}[th!]
    \centering
    \caption{Multiple Period Covariate Balance with Additional Covariates}
    \label{fig:multiple-period-covariate-balance-additional-covariates}
    \footnotesize
    \begin{subfigure}[b]{0.24\textwidth}
        \centering
        $\hat{\alpha} = 0.0672 \ (0.026)$ \smallskip
        \includegraphics[width=\textwidth]{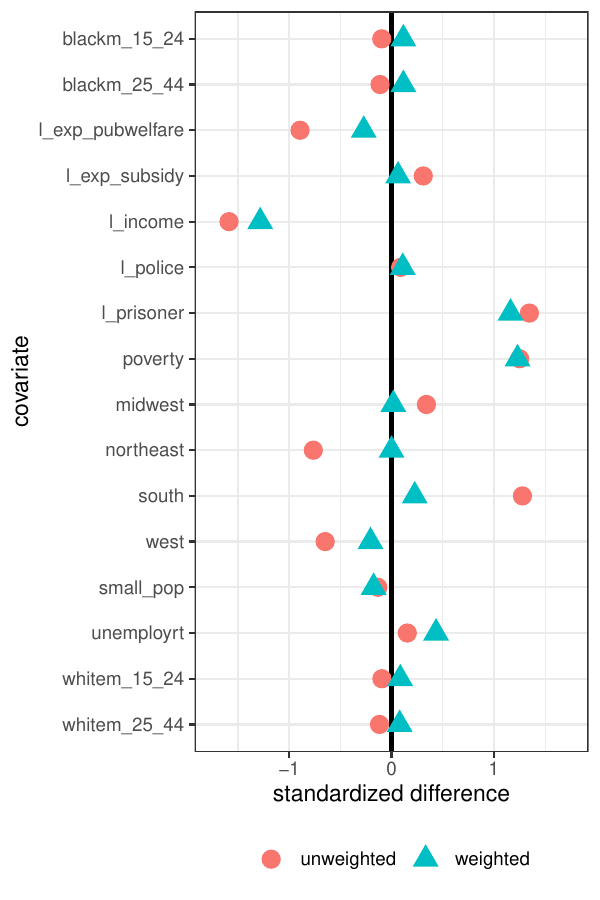}
        \vspace{-20pt}
        \caption{TWFE}
        \label{fig:sub1}
    \end{subfigure}
    \begin{subfigure}[b]{0.24\textwidth}
        \centering
        $\widehat{\mathrm{ATT}} = 0.106 \ (0.058)$ \smallskip
        \includegraphics[width=\textwidth]{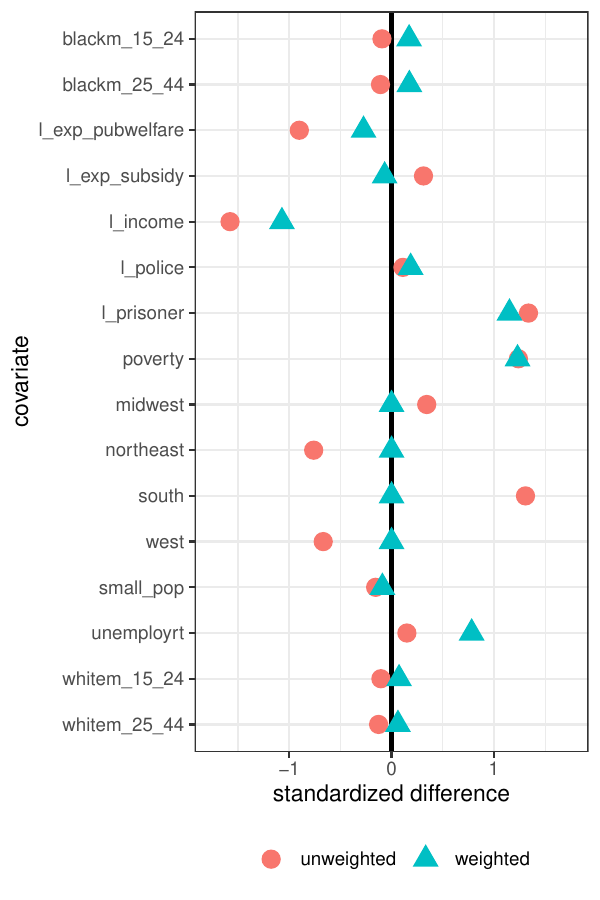}
        \vspace{-20pt}
        \caption{Reg.\,Adj.:\,$\Delta X, Z$}
        \label{fig:sub2}
    \end{subfigure}
    \begin{subfigure}[b]{0.24\textwidth}
        \centering
        $\widehat{\mathrm{ATT}} = 0.019 \ (0.043) $ \smallskip
        \includegraphics[width=\textwidth]{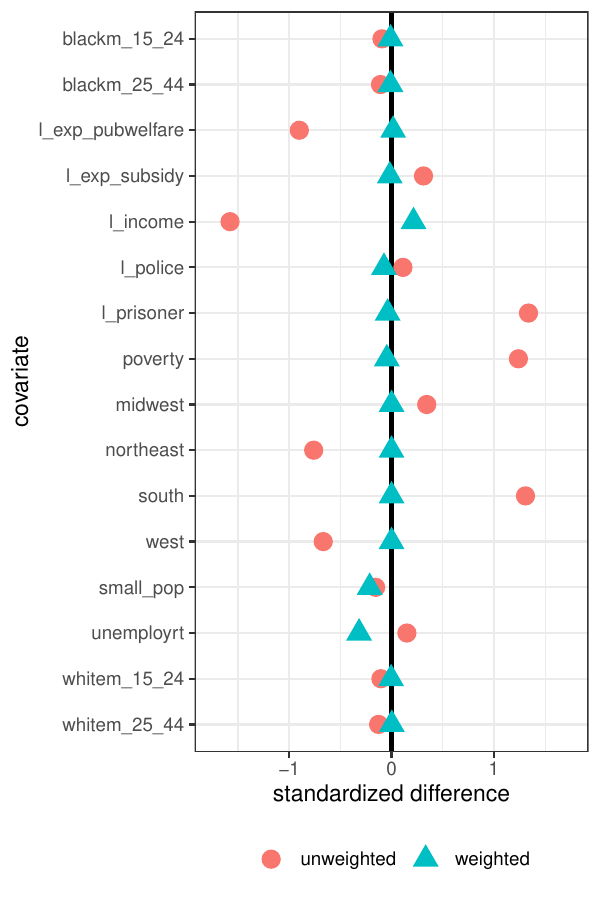}
        \vspace{-20pt}
        \caption{Reg.\,Adj.:\,$X_{g-1}, Z$}
        \label{fig:sub3}
    \end{subfigure}
    \begin{subfigure}[b]{0.24\textwidth}
        \centering
        $\widehat{\mathrm{ATT}} = 0.078 \ (0.180)$ \smallskip
        \includegraphics[width=\textwidth]{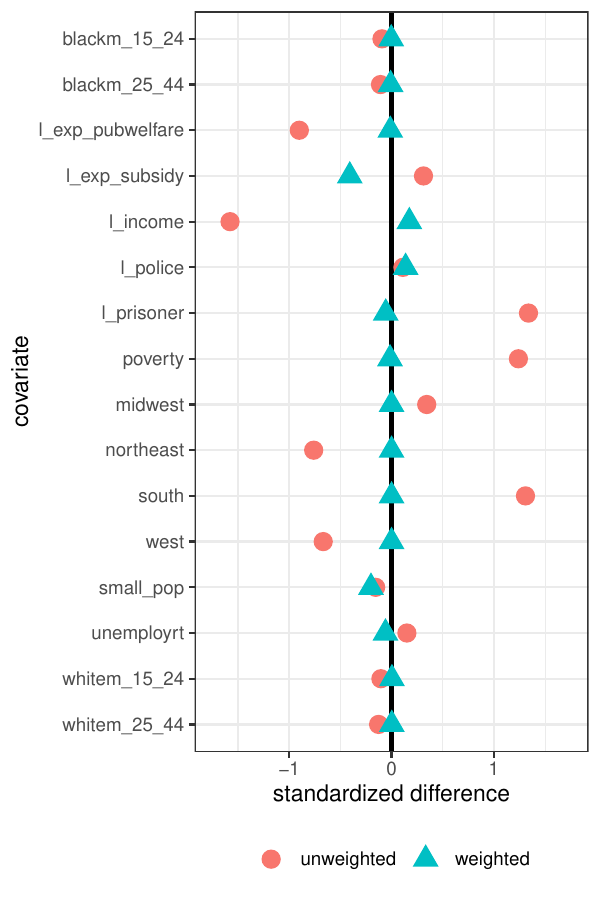}
        \vspace{-20pt}
        \caption{Reg.\,Adj.:\,$\Delta X$, $X_{g-1}, Z$}
        \label{fig:sub4}
    \end{subfigure}
    \begin{justify}
        \parbox{\textwidth}{ \normalsize \fontsize{12}{24}\selectfont \textit{Notes:} The figure reports estimates of the effects of stand-your-ground laws on homicides and covariate balance statistics using all available data from 2000-2010.  See the main text as well as \Cref{tab:ss} for a detailed explanation of each covariate. We also use state-specific average population as sampling weights, as in the main results in \citet{cheng-hoekstra-2013}.  The results in Panel (a) come from a TWFE regression that includes all the covariates listed in the figure. Panels (b)-(d) report regression adjustment results with different specifications for the covariates, as described in the main text.}

    \end{justify}
\end{figure}

\Cref{fig:multiple-period-covariate-balance-additional-covariates} provides the results.  Relative to the two-period results, there are larger differences across specifications.  The TWFE estimate is 0.067 (statistically significant).  Among the regression adjustment specifications, estimates vary considerably: the change-in-covariates specification (Panel (b)) gives a somewhat larger estimate than TWFE that is marginally significant; the level-of-covariates specification (Panel (c)) estimate is close to zero; and the specification with both levels and changes (Panel (d)) produces an estimate similar to TWFE but less precisely estimated.  Part of the difference between TWFE and regression adjustment arises because TWFE is affected by pre-treatment violations of parallel trends (the second term in \Cref{eqn:alpha-mp-diagnostics}).  Manually zeroing out this pre-treatment contribution increases the TWFE estimate by 31\% to 0.088.

The remaining differences are explained by different implicit weighting schemes.  In the figure, covariate balance is assessed in terms of how well the implicit weights across all post-treatment periods balance the average of each covariate. Relative to the raw data, the TWFE regression does improve covariate balance, though there are still some covariates that are severely unbalanced, particularly median income, log of incarceration rate, and poverty rate.  Regression adjustment that includes the change in covariates over time (Panel (b)) does not perform much better.  The last two specifications perform better in terms of covariate balance.  We also calculated the effective untreated group sample size across post-treatment periods for each regression adjustment specification (see Supplementary Appendix \ref{sa:ess-rem} for the specific calculation), finding values of 63.1, 26.7, and 9.9 for the specifications in Panels (b), (c), and (d), respectively.  The sharp drop for Panel (d) is reflected in its larger standard errors.  Taken together, the covariate balance diagnostics and effective sample sizes (both of which can be computed during the ``design phase'' of the analysis) provide an argument for preferring the specification in Panel (c) in this application.

\paragraph{Discussion}\mbox{}

A main takeaway from our application is that the functional form under which covariates enter the model is a first-order concern when controlling for particular covariates in the parallel trends assumption.  Approaches that inherit transformed covariates as a byproduct of the estimation strategy, whether it be TWFE, imputation/regression adjustment, or AIPW, performed poorly in terms of balancing the levels of time-varying covariates or time-invariant covariates.  On the other hand, regression adjustment and AIPW approaches that include \textit{any} level of a time-varying covariate and time-invariant covariates (such as the default implementation of \citet{callaway-santanna-2021}) performed substantially better.  In some cases, including time-invariant covariates and levels \textit{and} changes in time-varying covariates performed better, although this was not uniformly true.  We conjecture that a good heuristic for empirical work is to always include some version of the level of time-varying covariates (e.g., a pre-treatment value of each covariate) and time-invariant covariates. In applications with enough data, one should also consider including changes in time-varying covariates as well.

\section{Conclusion}\label{sec:conclusion}

We have considered difference-in-differences identification strategies when (i) the identification strategy hinges on comparing treated and untreated units with the same observed covariates and (ii) these covariates include time-varying and/or time-invariant variables.  In this empirically common setting, researchers have most often implemented this identification strategy using a TWFE regression as in \Cref{eqn:twfe}.  In this paper, we demonstrated a number of potential weaknesses of TWFE regressions in this context.  Some of these weaknesses, such as lack of robustness to multiple periods and variation in treatment timing, or reliance on certain linearity conditions, are likely not surprising given existing work in the difference-in-differences literature.  However, we also documented several other weaknesses that we termed \textit{hidden linearity bias}, which arise because the transformations used to eliminate the unit fixed effect in the TWFE regression also change the functional form of the covariates.  This transformation thus either effectively changes the identification strategy (to one where only the change in time-varying covariates is included in the parallel trends assumption) or relies heavily on a correctly specified linear model.  Empirical work rarely examines whether these conditions are reasonable, and, in most applications, they are likely to be considered too strong.  We proposed several diagnostic tools for assessing the sensitivity of TWFE regressions that include covariates to hidden linearity bias.  We also proposed an alternative estimation strategy, building on recent work in the DiD literature, that does not suffer from hidden linearity bias, does not require any auxiliary assumptions along the lines mentioned above, and is effectively no more complicated to implement in practice than the TWFE regression.

{
\setlength{\bibitemsep}{0\baselineskip}
\renewcommand*{\bibfont}{\normalsize}
\renewcommand{\bibsetup}{\fontsize{12}{24}\selectfont}
\raggedbottom
\printbibliography
\flushbottom
}

\appendix

\crefname{section}{Appendix}{Appendices}

\setcounter{proposition}{0}
\renewcommand{\theproposition}{A\arabic{proposition}}

\setlength{\topsep}{5pt}

\section{Proofs of Results with Two Periods} \label{app:proofs-two-periods}

\begin{lemma} \label{lem:lip} Under \Cref{ass:sampling,ass:overlap}, for $d \in \{0,1\}$,
\begin{align*}
    \E\Big[\L(D|\Delta X_{t^*}) \L_d(\Delta Y_{t^*} | \Delta X_{t^*}) \Big| D=d\Big] = \E\Big[\L(D|\Delta X_{t^*})\Delta Y_{t^*}\Big|D=d\Big].
\end{align*}
\end{lemma}
\noindent The proof of \Cref{lem:lip} is provided in the Supplementary Appendix.

\begin{lemma} \label{lem:twfe-denom} Under \Cref{ass:sampling,ass:overlap}, %
  \begin{align*}
    \E\Big[\big(D-\L(D|\Delta X_{t^*})\big)^2\Big] %
    &= \E\Big[ 1-\L(D|\Delta X_{t^*}) \Big| D=1 \Big] \pi.
  \end{align*}
\end{lemma}
\noindent The proof of \Cref{lem:twfe-denom} is provided in the Supplementary Appendix.

\Cref{lem:lip,lem:twfe-denom} are used in some of our arguments below involving linear projections.

\begin{proposition} \label{prop:twfe-lp} Under \Cref{ass:sampling,ass:overlap}, $\alpha$ from the regression in \Cref{eqn:fd} can be decomposed as
    \begin{equation*}
        \alpha = \E\left[ w(\Delta X_{t^*}) \Big(\L_1(\Delta Y_{t^*} | \Delta X_{t^*}) - \L_0(\Delta Y_{t^*} | \Delta X_{t^*})\Big) \Big| D=1\right],
    \end{equation*}
    where $w(\Delta X_t)$ are the same weights as in \Cref{thm:twfe}.
\end{proposition}

\vspace{-5pt}
\begin{proof}
    Starting with the numerator from \Cref{eqn:fwl}, we have that
    \begin{align}
        & \E\Big[\big(D - \L(D|\Delta X_{t^*})\big) \Delta Y_{t^*} \Big] \nonumber \\
        & \hspace{30pt} = \E\Big[\big(1-\L(D|\Delta X_{t^*})\big) \Delta Y_{t^*} \Big| D=1 \Big]\pi - \E\Big[ \L(D|\Delta X_{t^*}) \Delta Y_{t^*} \Big| D=0 \Big] (1-\pi) \nonumber \\
        & \hspace{30pt} = \E\Big[\big(1-\L(D|\Delta X_{t^*})\big) \L_1(\Delta Y_{t^*}|\Delta X_{t^*}) \Big| D=1 \Big]\pi \nonumber \\
        & \hspace{55pt} - \E\Big[ \L(D|\Delta X_{t^*}) \L_0(\Delta Y_{t^*}|\Delta X_{t^*}) \Big| D=0 \Big] (1-\pi) \nonumber \\
        & \hspace{30pt} = \E\Big[D\big(1-\L(D|\Delta X_{t^*})\big) \L_1(\Delta Y_{t^*}|\Delta X_{t^*}) \Big] - \E\Big[ (1-D)\L(D|\Delta X_{t^*}) \L_0(\Delta Y_{t^*}|\Delta X_{t^*}) \Big] \\
        & \hspace{30pt} = \E\left[D \big(1-\L(D|\Delta X_{t^*})\big) \Big(\L_1(\Delta Y_{t^*}|\Delta X_{t^*}) - \L_0(\Delta Y_{t^*} | \Delta X_{t^*}) \Big) \right] \nonumber \\
        & \hspace{55pt} + \E\left[ \big(D-\L(D|\Delta X_{t^*})\big) \L_0(\Delta Y_{t^*}|\Delta X_{t^*}) \right] \nonumber \\
        & \hspace{30pt} = \E\left[\big(1-\L(D|\Delta X_{t^*})\big) \Big(\L_1(\Delta Y_{t^*}|\Delta X_{t^*}) - \L_0(\Delta Y_{t^*} | \Delta X_{t^*}) \Big) \Big| D=1 \right]\pi, \label{eqn:prop-twfe-a}
    \end{align}
    where the first equality holds by the law of iterated expectations, the second equality holds by \Cref{lem:lip}, the third equality holds by applying the law of iterated expectations to both terms, the fourth equality holds by adding and subtracting $\E\Big[D\big(1-\L(D|\Delta X_{t^*})\big) \L_0(\Delta Y_{t^*}|\Delta X_{t^*})\Big]$, and the last equality holds by applying the law of iterated expectations for the first term and because
    \begin{align*}
        \E\left[ \big(D-\L(D|\Delta X_{t^*})\big) \L_0(\Delta Y_{t^*}|\Delta X_{t^*}) \right] = \E\left[ \big(D-\L(D|\Delta X_{t^*})\big) \Delta X_{t^*}' \right] \beta_0 = 0,
    \end{align*}
    where the first equality holds by the definition of $\L_0(\Delta Y_{t^*}|\Delta X_{t^*})$, and the second equality holds because $\Delta X_{t^*}$ is uncorrelated with the projection error $\big(D-\L(D|\Delta X_{t^*})\big)$.

    Combining the expression in \Cref{eqn:prop-twfe-a} with the expression for the denominator from \Cref{lem:twfe-denom} completes the proof, given the definition of the weights $w(\Delta X_{t^*})$.
\end{proof}
\vspace{-5pt}

\Cref{prop:twfe-lp} says that $\alpha$ is equal to a weighted average of the linear projection of the change in outcomes over time on the change in covariates over time for the treated group relative to the same linear projection for the untreated group.  Both the weights and the linear projection terms in the proposition depend only on linear projections, making them straightforward to compute in practice.  This proposition serves as an important intermediate step for our main results.

\begin{proposition}\label{prop:twfe-decomp2} Under \Cref{ass:sampling,ass:overlap}, $\alpha$ in \Cref{eqn:fd} can be decomposed as
    \hspace*{-100pt}
    \begin{align}
        \alpha &= \E\Big[ w(\Delta X_{t^*}) \Big( \E[\Delta Y_{t^*} | X_{t^*}, X_{t^*-1}, Z, D=1] - \E[\Delta Y_{t^*} | X_{t^*}, X_{t^*-1}, Z, D=0] \Big) \Big| D=1 \Big] \label{eqn:prop-twfe-alp1}\\
        &+ \E\Big[ w(\Delta X_{t^*}) \Big( \E[\Delta Y_{t^*} | X_{t^*}, X_{t^*-1}, Z, D=0] - \L_0(\Delta Y_{t^*} | \Delta X_{t^*}) \Big) \Big| D=1 \Big], \label{eqn:prop-twfe-alp2}
    \end{align}
    where $w(\Delta X_{t^*})$ are the same weights as in \Cref{thm:twfe}.
\end{proposition}

\vspace{-5pt}
\begin{proof}
    Starting from the numerator of the expression in \Cref{prop:twfe-lp}, we have that
    \begin{align}
        & \E\left[\big(1\smallminus\L(D|\Delta X_{t^*})\big) \Big(\L_1(\Delta Y_{t^*}|\Delta X_{t^*}) \smallminus \L_0(\Delta Y_{t^*} | \Delta X_{t^*}) \Big) \Big| D\smallequals 1 \right] \pi \nonumber \\
        & \hspace{5pt} = \E\left[\big(1\smallminus\L(D|\Delta X_{t^*})\big) \Big(\E[\Delta Y_{t^*} | X_{t^*}, X_{t^*-1}, Z, D\smallequals 1] - \E[\Delta Y_{t^*} | X_{t^*}, X_{t^*-1}, Z, D\smallequals 0] \Big) \Big| D\smallequals 1 \right] \pi \nonumber \\
        & \hspace{5pt} - \E\left[\big(1\smallminus\L(D|\Delta X_{t^*})\big) \Big\{ \Big(\E[\Delta Y_{t^*} | X_{t^*}, X_{t^*-1}, Z, D\smallequals 1] - \L_1(\Delta Y_{t^*}|\Delta X_{t^*})\Big)  \right. \label{eqn:twfe-prop-b} \\
        & \left. \hspace{105pt} - \Big(\E[\Delta Y_{t^*} | X_{t^*}, X_{t^*-1}, Z, D\smallequals 0] - \L_0(\Delta Y_{t^*} | \Delta X_{t^*}) \Big) \Big\} \Big| D\smallequals 1 \right] \pi, \nonumber
    \end{align}
    which holds by adding and subtracting \begin{align*}
        \E\left[\big(1-\L(D|\Delta X_{t^*})\big) \Big(\E[\Delta Y_{t^*} | X_{t^*}, X_{t^*-1}, Z, D=1] - \E[\Delta Y_{t^*} | X_{t^*}, X_{t^*-1}, Z, D=0] \Big) \Big| D=1 \right] \pi.
    \end{align*}
    Next, notice that
    \begin{align}
        \E\left[\big(1\smallminus\L(D|\Delta X_{t^*})\big) \E[\Delta Y_{t^*} | X_{t^*}, X_{t^*-1}, Z, D\smallequals 1] \Big| D\smallequals 1 \right] = \E\left[\big(1\smallminus\L(D|\Delta X_{t^*})\big) \Delta Y_{t^*} \Big| D\smallequals 1 \right], \label{eqn:twfe-prop-c}
    \end{align}
    which holds by the law of iterated expectations.  Further, notice that
    \begin{align}
        \E\left[\big(1-\L(D|\Delta X_{t^*})\big) \L_1(\Delta Y_{t^*} | \Delta X_{t^*}) \Big| D=1 \right] = \E\left[\big(1-\L(D|\Delta X_{t^*})\big) \Delta Y_{t^*} \Big| D=1 \right], \label{eqn:prop-twfe-d}
    \end{align}
    which holds by \Cref{lem:lip}.  That the terms in \Cref{eqn:twfe-prop-c,eqn:prop-twfe-d} are equal to each other implies that the first line of  \Cref{eqn:twfe-prop-b} is equal to 0.  Combining the second line of \Cref{eqn:twfe-prop-b} with the expression for the denominator in \Cref{eqn:fwl} from \Cref{lem:twfe-denom} completes the proof.
\end{proof}
\vspace{-5pt}

\Cref{prop:twfe-decomp2} decomposes $\alpha$ into two terms: (i) a weighted average of the average path of outcomes for the treated group relative to the path of outcomes for the untreated group (conditional on covariates), and (ii) a misspecification bias term that is a weighted average of the difference between the average path of outcomes for the untreated group conditional on time-varying and time-invariant covariates and the linear projection of $\Delta Y_{t^*}$ on $\Delta X_{t^*}$ for the untreated group.

\vspace{-5pt}
\begin{proof}[\textbf{Proof of \Cref{thm:twfe}}]
    First, it immediately follows from \Cref{ass:sampling,ass:overlap,ass:conditional-parallel-trends} that
    \\ $\mathrm{ATT}(X_{t^*},X_{t^*-1},Z) = \E[\Delta Y_{t^*} | X_{t^*},X_{t^*-1},Z,D=1] - \E[\Delta Y_{t^*} | X_{t^*},X_{t^*-1},Z,D=0]$.  This implies that \Cref{eqn:prop-twfe-alp1} in \Cref{prop:twfe-decomp2} is equal to $ \E[w(\Delta X_{t^*}) \mathrm{ATT}(X_{t^*},X_{t^*-1},Z) | D=1]$.  The second term comes from adding and subtracting terms to the expression in \Cref{eqn:prop-twfe-alp2} in \Cref{prop:twfe-decomp2}.  Note that the decomposition of this misspecification bias term is non-unique, as related terms could be added and subtracted in a different order (this is a similar property to many other decompositions).  The properties of the weights hold immediately by their definitions.
\end{proof}
\vspace{-5pt}

\vspace{-5pt}
\begin{proof}[\textbf{Proof of \Cref{thm:twfe-attX}}]
    The result holds immediately from \Cref{thm:twfe} by noticing that \Cref{ass:bias-elimination} directly implies that \Cref{eqn:thm-twfe-alp-a,eqn:thm-twfe-alp-b,eqn:thm-twfe-alp-c} are all equal to 0.
\end{proof}

\end{document}